\documentclass[twoside]{IEEEtran}

\usepackage{amsmath}
\usepackage{amssymb}
\usepackage{revsymb4-1}
\usepackage{amsfonts}
\usepackage{graphics}
\usepackage{pgf}
\usepackage{lipsum}
\usepackage{enumitem}

\usepackage{tikz}
\usetikzlibrary{positioning,shapes}

\newtheorem{theorem}{Theorem}

\newtheorem{definition}[theorem]{Definition}
\newtheorem{example}[theorem]{Example}
\newtheorem{proposition}[theorem]{Proposition}
\newtheorem{lemma}[theorem]{Lemma}
\newtheorem{remark}[theorem]{Remark}
\newenvironment{proof}[1][Proof]{\noindent\textbf{#1.} }{\ \rule{0.5em}{0.5em}}

\def\Tr{\operatorname{Tr}}
\def\tr{\operatorname{Tr}}
\def\id{\operatorname{id}}
\def\1{\openone}
\def\ox{\otimes}
\def\be{\begin{equation}}
\def\ee{\end{equation}}
\def\ba{\begin{eqnarray}}
\def\ea{\end{eqnarray}}

\newcommand{\ket}[1]{|#1 \rangle}
\newcommand{\bra}[1]{\langle #1|}
\newcommand{\proj}[1]{|#1 \rangle \! \langle #1 |}
\newcommand{\SWAP}{\operatorname{SWAP}}
\newcommand{\CC}{{\mathbb{C}}}

\newcommand{\EE}{{\mathbb{E}}}
\newcommand{\qed}{{\hfill\IEEEQED}}

\let\originalleft\left
\let\originalright\right
\renewcommand{\left}{\mathopen{}\mathclose\bgroup\originalleft}
\renewcommand{\right}{\aftergroup\egroup\originalright}

\usepackage{url}

\usepackage{color}

\begin{document}

\title{Quantum Channel Capacities \protect\\ with Passive Environment Assistance}

\author{Siddharth Karumanchi, \ Stefano Mancini, \thanks{S. Karumanchi and S. Mancini are with the
School of Science and Technology,
University of Camerino,
Via M. delle Carceri 9, I-62032 Camerino, Italy and
INFN--Sezione Perugia,
Via A. Pascoli, I-06123 Perugia, Italy. Email: stefano.mancini@unicam.it, siddharth.karumanchi@unicam.it.}
Andreas Winter, \thanks{A. Winter is with ICREA and
F\'{\i}sica Te\`{o}rica: Informaci\'{o} i Fen\`{o}mens Qu\`{a}ntics
Universitat Aut\`{o}noma de Barcelona,
ES-08193 Bellaterra (Barcelona), Spain. Email: andreas.winter@uab.cat.}
and Dong Yang\thanks{Dong Yang is with
F\'{\i}sica Te\`{o}rica: Informaci\'{o} i Fen\`{o}mens Qu\`{a}ntics
Universitat Aut\`{o}noma de Barcelona,
ES-08193 Bellaterra (Barcelona), Spain and
Laboratory for Quantum Information,
China Jiliang University,
Hangzhou, Zhejiang 310018, China. Email: dyang@cjlu.edu.cn.}
\thanks{Manuscript date 21 August 2014.}
}


\markboth{Karumanchi \MakeLowercase{\textit{et al.}}: Quantum Channel Capacities with Passive Environment Assistance}{Karumanchi \MakeLowercase{\textit{et al.}}: Quantum Channel Capacities with Passive Environment Assistance}

\maketitle

\begin{abstract}
  We initiate the study of \emph{passive environment-assisted}
  communication
  via a quantum channel, modeled as a unitary interaction between the
  information carrying system and an environment. In this model, the
  environment is controlled by a benevolent helper who can set its
  initial state such as to assist sender and receiver of the
  communication link. (The case of a malicious environment, also known
  as jammer, or arbitrarily varying channel, is essentially well-understood
  and comprehensively reviewed.)
  Here, after setting out precise definitions, focussing on the
  problem of quantum communication, we show that entanglement 
  plays a crucial role in this problem: indeed, the assisted 
  capacity where the helper is restricted to product states
  between channel uses is different from the one with unrestricted
  helper. Furthermore, prior shared entanglement between the
  helper and the receiver makes a difference, too.
\end{abstract}


\IEEEpeerreviewmaketitle

\section{Introduction}
\label{sec:intro}
\IEEEPARstart{I}{n} quantum Shannon theory it is customary to model communication
channels as completely positive and trace preserving (CPTP) maps on
states; this notion contains as a special case classical channels \cite{Wilde11}. 
It is a well-known fact that each CPTP map can be decomposed into a
unitary interaction with a suitable environment system and the discarding
of that environment. This means that the noise of the channel can be
entirely attributed to losing information into the environment,
which raises the question how much better one could communicate over
the channel if one had access to the environment. Note that ``access
to the environment'' is ambiguous at this point, but that
one can distinguish at least two broad directions, one concerned
with the exploitation of the information in the environment
after the interaction and the other with the control of the state 
of the environment before the interaction -- and of course both.

The first direction has been addressed starting from Gregoratti and Werner's 
``quantum lost and found''~\cite{GW03,GW04} and focusing on the error correction ability 
of this scheme for random unitary channels \cite{BCD05} as well as for other channel
types \cite{MCM11,MMM11}.
The problem was set in an information theoretic vein 
in~\cite{HK05} and culminated in the determination of
the ``environment-assisted'' quantum capacity of an interaction
with fixed initial state of the environment, but arbitrary measurements
on the environment output fed forward to the receiver~\cite{SVW05} 
(see Fig.~\ref{fig:modelactive}).
These findings were partially
extended to the classical capacity \cite{Winter07}, which revealed an interesting connection
to data hiding and highlighted the impact of the precise restriction on
the measurements on the combined channel-output and environment-output
system. Note that, whereas the usual capacity theory for quantum channels treats the
environment as completely inaccessible, these results assume full access to the environment
and classical communication to the receiver. 
Thus, whoever controls the environment can be considered as an \emph{active helper}.

\begin{figure}[ht]
\begin{center}
\begin{tikzpicture}[scale=0.5]
\draw [ultra thick] (2,-2) rectangle (6,2);
\draw[thick,red] (0,0) -- (2,0) ;
\draw [thick,blue](6,1) -- (10,1);
\draw [ultra thick] (10,0) rectangle (11,2) node[midway]{$U_{x}$};
\draw[thick, green] (6,-1) -- (8,-1);
\draw [ultra thick] (8,-2) rectangle (9,0) node[midway]{$M$};
\draw[->, line width = 2] (9,-1) -- (10.5,-0.1);
\draw[thick,blue] (11,1) -- (12,1);
\node at (9.9,-0.9) {$x$};
\node at (7,-1.5) {$F$};
\node[left] at (0,0){$A$}; \node[right] at (12,1){$B$};
\draw (4,0) node[font = \fontsize{40}{42}\sffamily\bfseries]{$\mathcal{N}$};
\end{tikzpicture}
\end{center}
\caption{Diagrammatic view of the three parties involved in the 
   communication setting with active helper.
   In this model the helper measures the output state with a POVM 
   $(M_{x})$, $\sum_x M_{x} = \1$, and sends the classical message $x$
   to Bob, who applies a corresponding unitary $U_{x}$ to recover the initial 
   message of Alice.}
  \label{fig:modelactive}
\end{figure}
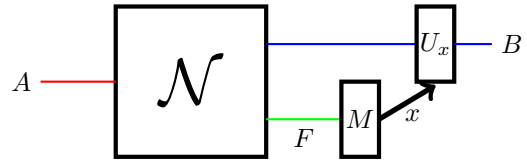

In the present paper we are concerned with the second avenue, to be precise a 
model where the communicating parties have no access to the environment-output
but instead there is a third party controlling the initial state of
the environment. The choice of initial environment state effectively
is a way of preparing a channel between Alice and Bob. Depending on the aim of
that party, we call the model communication with a \emph{passive helper}
if she is benevolent (because she only chooses the initial state and does 
not intervene otherwise), or communication in the presence of a \emph{jammer}
if he is malicious (see Fig.~\ref{fig:modelpassive}).

\begin{figure}[ht]
\begin{center}
\begin{tikzpicture}[scale=0.5]
\draw [thick, green](0,0) -- (1,0);
\draw[ultra thick] (1,-0.5) rectangle (2,0.5) node[midway]{$\eta$};
\draw[thick, green] (2,0) -- (4,0); \draw[thick, red] (2,2) -- (4,2); \draw[thick, blue] (8,1) -- (10,1);
\draw[ultra thick] (4,-1) rectangle (8,3);
\node[left] at (0,0){$H$}; \node[right] at (10,1){$B$};
\node[left] at (2,2){$A$}; \node[below] at (3,0){$E$};
\draw (6,1) node[font = \fontsize{40}{42}\sffamily\bfseries]{$\mathcal{N}$};
\end{tikzpicture}
\end{center}
  \caption{Diagrammatic view of the three parties involved in the communication 
    with a party controlling the environment input system.
    Depending on the goal of the party controlling the environment-input,
    either to assist or to obstruct the communication between the sender 
    Alice and the receiver Bob, we call it passive helper (Helen) or jammer
    (Jack), respectively.}
  \label{fig:modelpassive}
\end{figure}

In the next Section~\ref{sec:main} we shall define the model 
rigorously, as well as the different notions of assisted and 
adversarial codes and associated (quantum) capacities, and make 
initial general observations.
In Section~\ref{sec:two-qubit} we then go on to study 
two-qubit unitaries,which allow for the computation or estimation of capacities. They also show a range of general phenomena,
including super-activation of capacities that are discussed in Section ~\ref{sec:CA}.
These finding put into the focus a variation of the passive helper,
where she can use pre-shared entanglement with the receiver,
which model we explore in Section~\ref{sec:entangled}. 
We conclude in Section~\ref{sec:conclusion} with a number of
open problems and suggestions for future investigations.
Two appendices contain the technical details of the random coding 
capacity formula of the jammer model (Appendix~\ref{app:jammer}),
and the analysis of the (anti-)degradability properties of two-qubit 
unitaries (Appendix~\ref{app:two-qubit}).


\section{Assisted and adversarial capacities}
\label{sec:main}
As mentioned in the introduction we are concerned with the model of communication 
where there is a third party, other than the sender and receiver, who has access 
to the environment input system. The party's role is either to assist or hamper the 
quantum communication from Alice to Bob, which is distinguished in our nomenclature 
as Helen (helper) and Jack (jammer), respectively. 

Let $A$, $E$, $B$, $F$, etc.~be finite dimensional Hilbert spaces 
and $\mathcal{L}(X)$ denote the space of linear operators on the 
Hilbert space $X$.
Consider an isometry $V:A\ox E \hookrightarrow B\ox F$, which
defines the channel (CPTP map) 
$\mathcal{N}:\mathcal{L}(A\ox E) \rightarrow \mathcal{L}(B)$,
whose action on the input state is
\begin{equation*}
  \mathcal{N}^{AE \rightarrow B}(\rho) = \tr_F V \rho V^\dag.
\end{equation*}
The complementary channel, 
$\widetilde{\mathcal N}:\mathcal{L}(A\ox E) \rightarrow \mathcal{L}(F)$, 
is given by  
\begin{equation*}
  \widetilde{\mathcal N}^{AE \rightarrow F}(\rho) = \tr_B V \rho V^\dag.
\end{equation*}

By inputting an environment state $\eta$ on $E$, an effective channel 
$\mathcal{N}_\eta:\mathcal{L}(A) \rightarrow \mathcal{L}(B)$ is defined, via
\[
  \mathcal{N}_\eta^{A\rightarrow B}(\rho) = \mathcal{N}^{AE\rightarrow B}(\rho\ox\eta).
\]
Clearly, for channels $\mathcal{N}_i:\mathcal{L}(A_i E_i) \rightarrow \mathcal{L}(B_i)$
and states $\eta_i$,
\[
  (\mathcal{N}_1\ox\mathcal{N}_2)_{\eta_1\ox\eta_2} 
        = (\mathcal{N}_1)_{\eta_1} \ox (\mathcal{N}_2)_{\eta_2}.
\]
Note that if $\eta$ is pure, then the complementary channel is given by
\[
  \widetilde{\mathcal{N}_\eta} = (\widetilde{N})_\eta,
\]
but this is not true in general for mixed states $\eta$.

\begin{figure}[ht]
\begin{tikzpicture}[scale=0.28]
\draw [ultra thick] (5,-11) rectangle (10,-1) ;
\draw [ultra thick] (18,-11) rectangle (23,-1) ;
\draw [ultra thick] (13,-3) rectangle (15,-1) node[midway]{$\mathcal{N}$};
\draw [ultra thick] (13,-6) rectangle (15,-4) node[midway]{$\mathcal{N}$};
\draw [ultra thick] (13,-11) rectangle (15,-9) node[midway]{$\mathcal{N}$};
\draw [thick, red] (0,0)   -- (3,-6) -- (5,-6);
\draw [thick, red] (10,-1.5) -- (13,-1.5);
\draw [thick, red] (10,-4.5) -- (13,-4.5);
\draw [thick, red] (10,-9.5) -- (13,-9.5);
\draw [thick,blue] (15,-2) -- (18,-2);
\draw [thick, blue] (15,-5) -- (18,-5);
\draw [thick, blue] (15,-10) -- (18,-10);
\draw [thick, blue] (23,-6) -- (25,-6) -- (28,0);
\draw [thick, purple] (0,0) -- (3,6) -- (25,6) -- (28,0);
\draw [loosely dashed] (-0.5,0) -- (28.5,0);
\draw [thick, green] (11,-13) -- (11,-2.5) -- (13,-2.5);
\draw [thick, green] (11,-13) -- (11.2,-5.5) -- (13,-5.5);
\draw [thick, green] (11,-13) -- (11.4,-10.5) -- (13,-10.5);
\node[above] at (4,-6) {$A_{0}$}; \node[above] at (24,-6){$B_{0}$};\node[below] at (14,6){$R$};
\node[left] at (0,0){$\Phi$};
\node[right] at (28,0){$\sigma$};
\node[above] at (10.5,-1.5){$A$};\node[above] at (10.5,-4.5){$A$};\node[above] at (10.5,-9.5){$A$};
\node[below] at (12.3,-2.5){$E$};\node[below] at (12.3,-5.5){$E$};\node[below] at (12.3,-10.5){$E$};\node[above] at (17,-2){$B$};\node[above] at (17,-5){$B$};\node[above] at (17,-10){$B$};
\node[below] at (11,-13){$\eta$};
\draw (7.5,-6) node[font = \fontsize{40}{42}\sffamily\bfseries]{$\mathcal{E}$};
\draw (20.5,-6) node[font = \fontsize{40}{42}\sffamily\bfseries]{$\mathcal{D}$};
\draw [thick,dotted] (10.5,-6.5) -- (10.5,-7.5);
\draw[thick,dotted] (12.3,-7.5) -- (12.3,-8.5);
\draw[thick,dotted] (14,-7) -- (14,-8);
\draw[thick,dotted] (17,-7) -- (17,-8);
\end{tikzpicture}
\caption{Schematic of a general protocol to transmit quantum 
    information with passive assistance from the environment; $\mathcal{E}$ 
    and $\mathcal{D}$ are the encoding and decoding maps respectively, 
    the initial state of the environment is $\eta$.}
\label{fig:infotask}
\end{figure}

Referring to Fig.~\ref{fig:infotask}, to send information down this channel from
Alice to Bob, we furthermore need an encoding CPTP map 
$\mathcal{E}:\mathcal{L}(A_{0}) \rightarrow \mathcal{L}(A^n)$
and a decoding CPTP map $\mathcal{D}:\mathcal{L}(B^n) \rightarrow \mathcal{L}(B_0)$,
where the dimension of $A_0$ is equal to the dimension of $B_0$. 
The output after the overall dynamics, when we input a maximally entangled
test state $\Phi^{RA_0}$, with $R$ being the inaccessible reference system, is 
$\sigma^{RB_0} 
  = \mathcal{D}\bigl(\mathcal{N^{\ox}}\bigl(\mathcal{E}(\Phi^{RA_0}) \ox \eta^{E^n}\bigr)\bigr)$.

\medskip
\begin{definition} 
  \label{def:passive-assisted-code}
  A \emph{passive environment-assisted quantum code} of block length $n$ is a triple 
  $(\mathcal{E}^{A_0 \rightarrow A^n},\eta^{E^n},\mathcal{D}^{B^n \rightarrow B_0})$. 
  Its \emph{fidelity} is given by $F = \Tr \Phi^{RA_0} \sigma^{RB_0}$,
  and its \emph{rate} $\frac{1}{n}\log|A_0|$. 
  
  A rate $R$ is called \emph{achievable} if there are codes of all block lengths
  $n$ with fidelity converging to $1$ and rate converging to $R$. The 
  \emph{passive environment-assisted quantum capacity} of $V$, denoted
  $Q_H(V)$, or equivalently $Q_H(\mathcal{N})$, is the maximum achievable rate.
  
  If the helper is restricted to fully separable states $\eta^{E^n}$, i.e.~convex
  combinations of tensor products $\eta^{E^n} = \eta_1^{E_1} \ox \cdots \ox \eta_n^{E_n}$,
  the largest achievable rate is denoted $Q_{H\ox}(V) = Q_{H\ox}(\mathcal{N})$.
\end{definition}

\medskip
A very similar model, however with the aim of maximizing the ``transfer fidelity''
(averaged over all pure states of $A$), was considered recently by 
Liu \emph{et al.}~\cite{LGZ13}. Although the figure of merit is different,
the objective of that paper is, like ours, a quantitative index for the 
transmission power of a bipartite unitary, assisted by a benevolent helper.

As the fidelity is linear in the environment state $\eta$, without loss of
generality $\eta$ may be assumed to be pure, both for the unrestricted and
separable helper. We shall assume this from now on always in the helper
scenario, without necessarily specifying it each time.

\medskip
\begin{remark}
  Our model, since it allows for an isometry $V$, includes the plain 
  Stinespring dilation $V: A \hookrightarrow B\ox F$ of a quantum channel
  (CPTP map) $\mathcal{N}:\mathcal{L}(A) \rightarrow \mathcal{L}(B)$,
  for trivial ($1$-dimensional) $E=\CC$ so that the helper doesn't really 
  have any choice of initial state. In this case the quantum capacity is
  well-understood thanks to the works of Schumacher, Lloyd, Shor and
  Devetak. The fundamental quantity is the 
  \emph{coherent information}~\cite{Schu96,NS96,BNS98}, see also~\cite{Wilde11}
  \[
    I(A\rangle B)_\sigma := S(\sigma^B)-S(\sigma^{AB}) = -S(A|B)_\sigma,
  \]
  which needs to be evaluated for states 
  $\sigma^{AB} = (\id\ox\mathcal{N}^{A'\rightarrow B})\phi^{AA'}$, where
  $\phi^{AA'}$ is a purification of a generic density matrix $\rho^A$:
  \[
    I_c(\rho;\mathcal{N}) := I(A\rangle B)_{(\id\ox\mathcal{N})\phi}
                           = S(\mathcal{N}(\rho)) - S(\widetilde{\mathcal{N}}(\rho)).
  \]
  Then~\cite{Schu96,NS96,BNS98,Lloyd96,Shor00,Devetak03},
  \[
    Q(\mathcal{N}) = \sup_n \max_{\rho^{(n)}} \frac1n I_c\bigl(\rho^{(n)};\mathcal{N}^{\ox n}\bigr),
  \]
  where the maximum is over all states $\rho^{(n)}$ on $A^n$. It is known that
  the supremum over $n$ (the ``regularization'') is necessary~\cite{SS96,DSS98},
  except for some special channels -- see below.
  
  On the other hand, the helper has the largest range of options to assist 
  if $V$ is a \emph{unitary}. This will be the case that shall occupy us 
  most in the sequel. However, in any case, we assume that the input to
  $V$ is a product state between Alice and Helen, since they have to
  act independently, albeit in coordination.
\end{remark}

\medskip
Before we continue with our development of the theory of passive
environment-assisted capacities, we pause for a moment to reflect on
the role of the environment. While our above definitions model a
benevolent agent controlling the environment input, one may ask
what results if instead he is \emph{malevolent}, i.e.~trying to jam
the communication between Alice and Bob. This is captured by the following 
definition:

\medskip
\begin{definition}
  A \emph{quantum code} of block length $n$ for the jammer channel 
  $\mathcal{N}^{AE\rightarrow B}$ is a pair 
  $(\mathcal{E}^{A_0 \rightarrow A^n},\mathcal{D}^{B^n \rightarrow B_0})$,
  with two spaces $A_0$ and $B_0$ of the same dimension.
  Its rate is, as before, $\frac{1}{n} \log|A_0|$, while the fidelity is given by
  \[
    F := \min_{\eta^{E^n}} \Tr \Phi^{RA_0} \sigma^{RB_0},
  \]
  where $\eta^{E^n}$ ranges over all states on $E^n$, and
  $\sigma^{RB_0} 
    = \mathcal{D}\bigl(\mathcal{N^{\ox}}\bigl(\mathcal{E}(\Phi^{RA_0}) \ox \eta^{E^n}\bigr)\bigr)$, 
  with a maximally entangled state $\Phi^{RA_0}$.
  
  A \emph{random quantum code} is given by an ensemble of codes
  $(\mathcal{E}_\lambda^{A_0 \rightarrow A^n},\mathcal{D}_\lambda^{B^n \rightarrow B_0})$
  with a random variable $\lambda$. The rate is as before, and the fidelity
  \[
    \overline{F} := \min_{\eta^{E^n}} \EE_\lambda \Tr \Phi^{RA_0} \sigma_\lambda^{RB_0},
  \]
  where now 
  $\sigma_\lambda^{RB_0} 
   = \mathcal{D}_\lambda\bigl(\mathcal{N^{\ox}}\bigl(\mathcal{E}_\lambda(\Phi^{RA_0}) 
                                                              \ox \eta^{E^n}\bigr)\bigr)$.
  
  The corresponding \emph{adversarial quantum capacities}, 
  to emphasize the presence of the jammer, 
  are denoted $Q_J(\mathcal{N})$ and $Q_{J,r}(\mathcal{N})$, respectively.
\end{definition}

\medskip
\begin{remark}
  The special case where the jammer controls a classical input $E$, i.e.~there
  is an orthonormal basis $\{\ket{s}\}$ of $E$ such that
  \[
    \mathcal{N}(\rho \ox \ket{s}\!\bra{t}) = \delta_{st}\,\mathcal{N}_s(\rho),
  \]
  has been introduced and studied in-depth by Ahlswede \emph{et al.}~\cite{ABBN}
  under the name of \emph{arbitrarily varying quantum channel (AVQC)}. In
  other words, there the communicating parties are controlling genuine quantum
  systems (naturally, as they are supposed to transmit quantum information),
  whereas the jammer effectively only has a classical choice $s$. 
  
  Our model here lifts this restriction and generalizes the AVQC to a fully 
  quantum jammer channel. 
  This has the very important consequence that the jammer now can
  choose to prepare channels for Alice and Bob that are not tensor products
  of $n$ single-system channels, or convex combinations thereof, but have
  other, more subtle noise correlations between the $n$ systems.
\end{remark}

\medskip
It turns out that the worst behaviour of the jammer, at least in the
random code case, is to choose one, pessimal, environment input to $\mathcal{N}$
and use it in all $n$ instances. The following theorem is proved in
Appendix~\ref{app:jammer}.

\medskip
\begin{theorem}
  \label{thm:QJ}
  For any jammer channel $\mathcal{N}^{AE \rightarrow B}$, 
  \[
    Q_{J,r}(\mathcal{N}) 
      = \sup_n \max_{\rho^{(n)}} \min_{\eta} \frac1n I_c\bigl(\rho^{(n)};(\mathcal{N}_\eta)^{\ox n}\bigr),
  \]
  where the maximization is over states $\rho^{(n)}$ on $A^n$, and the minimization
  is over arbitrary (mixed) states $\eta$ on $E$.
\end{theorem}

\medskip
See~\cite{ABBN} and~\cite{BN13} for a detailed discussion of the role
of shared randomness in the theory of the AVQC model; these authors
suggest that $Q_J = Q_{J,r}$ for all jammer channels, at least for all AVQCs, 
which however should be contrasted with the findings of~\cite{BN14} that there
are AVQCs for which the \emph{classical} capacity assisted by shared randomness
is positive while without that resource it is zero.

\bigskip
Let us now resume our discussion of environment-assisted quantum capacity,
deriving capacity theorems analogous to the one above for the jammer model. 
For the latter we saw that (mixed) product states are asymptotically optimal for
the jammer. It will turn out that restricting the helper to product (separable) 
states can be to severe disadvantage; while from the definitions, for any 
isometry $V$ we have $Q_{H\ox}(V) \leq Q_{H}(V)$, the inequality can be strict.

\medskip
\begin{theorem}
  \label{QH+QHtens}
  For an isometry $V:AE \longrightarrow BF$, the passive environment-assisted
  quantum capacity is given by 
  \begin{equation}\begin{split}
    \label{eq:QH}
    Q_{H}(V) &= \sup_{n} \max_{\eta^{(n)}} \frac1n Q(\mathcal{N}^{\ox n}_{\eta^{(n)}}) \\
             &= \sup_{n} \max_{\rho^{(n)},\eta^{(n)}}
                   \frac{1}{n} I_c\bigl( \rho^{(n)};\mathcal{N}^{\ox n}_{\eta^{(n)}} \bigr),
  \end{split}\end{equation}
  where the maximization is over states $\rho^{(n)}$ on $A^n$ and \emph{pure}
  environment input states $\eta^{(n)}$ on $E^n$.
  
  Similarly, the capacity with separable helper is given by the same formula,
  \begin{equation}\begin{split}
    \label{eq:QHtens}
    Q_{H\ox}(V) &= \sup_{n} \max_{\eta^{(n)}=\eta_1\ox\cdots\ox\eta_n}
                     \frac1n Q(\mathcal{N}_{\eta_1}\ox\cdots\ox\mathcal{N}_{\eta_n})   \\
                &= \sup_{n} \max_{\rho^{(n)},\eta^{(n)}}
                     \frac{1}{n} I_c\bigl( \rho^{(n)};\mathcal{N}^{\ox n}_{\eta^{(n)}} \bigr),
  \end{split}\end{equation}
  but now varying only over (pure) product states, 
  i.e.~$\eta^{(n)} = \eta_1 \ox \cdots \ox \eta_n$.
  
  As a consequence, $Q_H(V) = \lim_{n\rightarrow\infty} \frac1n Q_{H\ox}(V^{\ox n})$.
\end{theorem}
\begin{proof}
  The direct parts, i.e.~the ``$\geq$'' inequality, follows directly from the 
  Lloyd-Shor-Devetak (LSD) theorem~\cite{Lloyd96,Shor00,Devetak03}, applied to the channel
  $(\mathcal{N}^{\ox n})_{\eta^{(n)}}$, to be precise asymptotically many
  copies of this block-channel, so that the i.i.d.~theorems apply (cf.~\cite{Wilde11}).

  For the converse (i.e.~``$\leq$''), we apply directly the argument of
  Schumacher, Nielsen and Barnum~\cite{Schu96,NS96,BNS98}: Consider a code
  of block length $n$ and fidelity $F$, where the helper uses an environment
  state $\eta^{(n)}$; otherwise we use notation as in Fig.~\ref{fig:infotask}.
  Then, first of all, $\frac12 \|\sigma-\Phi\|_1 \leq \sqrt{1-F} =: \epsilon$,
  cf.~\cite{Fuchs-vandeGraaf}. Now, Fannes' inequality~\cite{Fannes} 
  can be applied, at least once $2\epsilon \leq \frac1e$ 
  (i.e.~when $F$ is large enough), yielding
  \[\begin{split}
    I(R\rangle B_0)_\sigma &=    S(\sigma^{B_0}) - S(\sigma^{RB_0}) \\
                           &\geq S(\sigma^{B_0}) \\
                           &\geq S(\Phi^{A_0}) - 2\epsilon \log|B_0| - H_2(2\epsilon) \\
                           &\geq (1-2\epsilon) \log|A_0| - 1.
  \end{split}\]
  On the other hand, with $\omega = (\id\ox\mathcal{E})\Phi$,
  \[\begin{split}
    I(R\rangle B_0)_\sigma 
             &\leq I(R\rangle B^n)_{(\id\ox\mathcal{N}^{\ox n}_{\eta^{(n)}})\omega} \\
             &\leq \max_{\ket{\phi}^{RA^n}} I(R\rangle B^n)_{(\id\ox\mathcal{N}^{\ox n}_{\eta^{(n)}})\phi} \\
             &=    \max_{\rho^{(n)}} I_c\bigl( \rho^{(n)};(\mathcal{N}^{\ox n})_{\eta^{(n)}} \bigr),
  \end{split}\]
  using first data processing of the coherent information and then its
  convexity in the state~\cite{NS96}.
  As $n \rightarrow \infty$ and $F\rightarrow 1$, the upper bound on the rate
  follows -- depending on $Q_H$ or $Q_{H\ox}$, without or with restrictions
  on $\eta^{(n)}$.
\end{proof}

\medskip
\begin{remark}
  The channels $\mathcal{N}:\mathcal{L}(A\ox E) \longrightarrow \mathcal{L}(B)$
  can equivalently be seen as (two-sender-one-receiver) \emph{quantum multi-access
  channels}. These channels were introduced and studied in~\cite{W01,YHD08}
  under the aspect of characterizing their \emph{capacity region} of all
  pairs or rates $(R_A,R_E)$ at which the users, Alice and Helen, controlling the two
  input registers can communicate with Bob. In fact, while in~\cite{W01} only
  special channels and classical communication were considered, Ref.~\cite{YHD08}
  extended this to general CPTP maps and the consideration of quantum
  communication.
  
  Clearly, knowing the capacity region for some $\mathcal{N}^{AE\rightarrow B}$
  implies the environment-assisted capacity:
  \[
    Q_H(\mathcal{N}) = \max \{ R : (R,0) \in \text{ capacity region} \}.
  \]
  Unfortunately, however, in general only a regularized capacity formula
  is available, much like our Theorem~\ref{QH+QHtens}. Thus, the general
  multi-access viewpoint does not seem to help particularly with the computation
  of $Q_H$ or $Q_{H\ox}$.
\end{remark}

\medskip
\begin{proposition}
  \label{prop:asymp-continuity}
  The capacities $Q_H$, $Q_{H\ox}$ and $Q_{J,r}$ are continuous in the
  channel, with respect to the diamond (or completely bounded) norm. 
  Concretely, if $\left\|\mathcal{N}- \mathcal{M} \right\|_{\diamond} \leq \epsilon$, 
  then
  \begin{align*}
   \bigl| Q_{H\ox}(\mathcal{N})-Q_{H\ox}(\mathcal{M}) \bigr| &\leq 8\epsilon\log|B| + 4 H_2(\epsilon), \\
   \bigl| Q_{H}(\mathcal{N})-Q_{H}(\mathcal{M}) \bigr|       &\leq 8\epsilon\log|B| + 4 H_2(\epsilon), \\
   \bigl| Q_{J,r}(\mathcal{N})-Q_{J,r}(\mathcal{M}) \bigr|   &\leq 8\epsilon\log|B| + 4 H_2(\epsilon), \\
  \end{align*}
  with the binary entropy $H_2(x) = -x\log x -(1-x)\log(1-x)$.
\end{proposition}
\begin{proof}
This is essentially the argument of Leung and 
Smith~\cite[Thm.~6; Lemma 1; Cor.~2]{LS09}. 
We can apply this because we have the formulas for these capacities in terms of coherent 
informations $\frac{1}{n} I_c\bigl( \rho^{(n)};\mathcal{N}^{\ox n}_{\eta^{(n)}} \bigr)$,
according to Theorem~\ref{QH+QHtens}.
The only new ingredient is that now the parameter is the joint input state
$\rho^{(n)} \ox \eta^{(n)}$, but fixing that the proof via the ``hybrid argument''
in~\cite{LS09} goes through. 
\end{proof}

\medskip
We remark here that it is not known at the time of writing, whether
$Q_J$ is continuous in the channel, a problem that is in fact closely 
tied to the question whether $Q_J=Q_{J,r}$ for all channels.

\bigskip
Given that in our formulation of the environment-assisted quantum capacity,
the ordinary quantum channel capacity is contained as a special case, 
it is clear
that we cannot make many general statements about either $Q_H$ or $Q_{H\ox}$.
However, focusing from now on on unitaries $V:AE \longrightarrow BF$, we
will in the sequel explore the assisted capacities by looking at specific
classes of interactions which exhibit interesting or even unexpected 
behaviour.

To start, what are the unitaries $V:AE \longrightarrow BF$, say with
equal dimensions of $A$ and $B$, with maximal capacity $\log|B|$? For
$Q_H(V)$ this seems a non-trivial question, but for $Q_{H\ox}(V)$,
invoking the result of~\cite{BEHY11}, we find that
$Q_{H\ox}(V)=\log|B|$ if and only if there exist states $\ket{\eta}\in E$,
$\ket{\phi}\in F$, and a unitary $U:A \longrightarrow B$ such that
\[
  V(\ket{\psi}^A\ket{\eta}^E) = (U\ket{\psi})^B\ket{\phi}^F,
\]
which in principle can be checked algebraically. In other words,
in this case, one of the channels $\mathcal{N}_\eta$ induced by
choosing an environment input state is the conjugation by a unitary.
In the search for non-trivial channels, we find the following result.

\begin{theorem}
  \label{thuni}
  Let $|A|=|B| = 2$, $|E|=|F| = d\le 4$ and consider $d$ linearly independent unitaries $U^{A\rightarrow B}_k\in{\rm U}(2)$. If the unitary $V: AE\longrightarrow BF$ is such that it induces a mixture of conjugation by $U_k$'s  
for any state $|\eta\rangle \in E$, 
then $V$ is a controlled-unitary gate:
\[
  V^{AE\rightarrow BF} = \sum_k U_k^{A\rightarrow B} \otimes |{f_k}\rangle^F\langle{e_k}|^E,
\]
with suitable orthonormal bases 
$\{\ket{e_k}\}_k$ and $\{\ket{f_k}\}_k$ of $E$ and $F$, respectively.
\end{theorem}
\begin{proof}
Let us start from the requirement that $V$ gives rise to mixture of conjugation by $U_k$'s 
in the states $\{\ket{j}\}_j$ of a basis of the environment $E$. 
W.l.o.g.~we can write the action of $V$ as follows
\begin{equation}
\label{V}
  V |\psi\rangle^A |j\rangle^E=\sum_k U_k |\psi\rangle |v_{jk}\rangle,
\end{equation}
where $|v_{jk}\rangle$ are non-normalized states of $E$. Then, 
let us consider a standard maximally entangled state $\ket{\Psi}^{RA}$
between a reference system $R$ and the input system $A$. We have
\begin{equation*}
(I \otimes V) |\Psi\rangle^{RA} |j\rangle^E=\sum_k (I\otimes U_k) |\Psi\rangle |v_{jk}\rangle
=: \sum_k |\Psi_k\rangle |v_{jk}\rangle,
\end{equation*}
with all the $|\Psi_k\rangle^{RA}$ maximally entangled states. 
The trace over $E$ gives the Choi-Jamiolkowski state of the channel which in turn must represent 
a mixture of conjugations by $U_k$'s, hence the following equality must hold true:
\begin{equation*}
\sum_{kk'} \langle v_{jk'} | v_{jk} \rangle |\Psi_k\rangle\langle\Psi_{k'}|
=\sum_k p_k |\Psi_k\rangle\langle\Psi_k|,
\end{equation*} 
for some probability distribution $\{p_k\}_k$.
Since the $\ket{\Psi_k}$ are linearly independent (as a consequence of the linear 
independence of the unitaries $U_k$), 
we necessarily must have vanishing scalar products 
$\langle v_{jk} | v_{j k'}\rangle=0$ for all $j$ and all $k\neq k'$. 

For a generic environment state $|\eta\rangle=\sum_j \eta_j |j\rangle$  it is
\begin{equation}
\label{Veta}
V |\psi\rangle \sum_j \eta_j |j\rangle=\sum_k U_k |\psi\rangle \sum_j \eta_j |v_{jk}\rangle,
\end{equation}
and using the same argument as above we end up
with the requirement that the states $\{\sum_j\eta_j |v_{jk}\rangle\}_k$ have be orthogonal (for different values 
of $k$). 
Actually this must be true for any value of the $\eta_j$s, hence the only possibility is that the vectors 
$|v_{jk}\rangle$ result as $|v_{jk}\rangle=c_{jk} | f_k\rangle$ with $\{| f_k\rangle \}_k$ orthonormal.

This can be proved by considering the scalar product between 
\begin{equation*}
\sum_j\eta_j |v_{j r}\rangle \quad {\rm and} \quad \sum_j\eta_j |v_{j s}\rangle,
\end{equation*}
(for arbitrary values $r\neq s$)
with all $\eta_j=0$ except $\eta_m$ and $\eta_n$ 
(for any values $m\neq n$), which yield the following conditions:
\begin{equation*}
\left( \overline{\eta}_m \langle v_{m r}| + \overline{\eta}_n \langle v_{nr} |\right)
\left(\eta_m | v_{ms} \rangle + \eta_n |v_{ns}\rangle  \right)=0.
\end{equation*}
Then, we may notice that
\begin{align*}
  \eta_m &= \eta_n=1 \ \Rightarrow\ \langle v_{mr} |v_{ns}\rangle= - \langle v_{nr} |v_{ms}\rangle, \\
  \eta_m &= \eta_n=i \ \Rightarrow\ \langle v_{mr} |v_{ns}\rangle=  \langle v_{nr} |v_{ms}\rangle.
\end{align*}
To simultaneously satisfy these conditions it must hold that
$ \langle v_{mr} |v_{ns}\rangle=  \langle v_{nr} |v_{ms}\rangle=0$.
Due to the arbitrariness of $r,s,m,n$ we can conclude that 
$\langle v_{jk}| v_{j'k'}\rangle=0$ for $k\neq k'$ and for any $j,j'$, 
i.e.~$|v_{jk}\rangle=c_{jk} | f_k\rangle$ with $\{| f_k\rangle \}_k$ orthonormal.

Thus, the action \eqref{Veta} of $V$ in the environment basis states 
$\{|j\rangle\}_j$  will result as
\begin{equation*}
  V |\psi\rangle  |j\rangle= \sum_k U_k |\psi\rangle  c_{jk} | f_{k}\rangle.
\end{equation*}
Therefore, in the basis $\{|j\rangle\}_j$ the unitary $V$ can be written as
\begin{equation*}
V= \sum_{j,k} U_k \otimes  c_{jk} | f_{k}\rangle\langle j |
                = \sum_k U_k \otimes  | f_{k}\rangle \langle e_k |,
\end{equation*}
where we have defined the vectors
\begin{equation*}
  |e_k\rangle:=\sum_j \overline{c}_{kj} |j\rangle.
\end{equation*}

Finally using the condition 
\begin{equation*}
  \sum_k \langle v_{j k} | v_{j'k} \rangle=\delta_{jj'}
\end{equation*}
coming from the unitarity of $V$, we have
\begin{equation*}
\sum_{j} c_{jk} \overline{c}_{k'j} =\delta_{kk'},
\end{equation*} 
expressing to the orthonormality of $\{|e_k\rangle\}_k$.
\end{proof}

\medskip
We conjecture furthermore that for $|A|=|B| = 2$ and $|E|=|F| = d$
arbitrary, if $V: AE\longrightarrow BF$ is such that it induces random-unitary
(equivalently: unital~\cite{KR01}) channels $\mathcal{N}_\eta$
for all states $\ket{\eta} \in E$, then $V$ is essentially a controlled-unitary
gate:
\[
  V^{AE\rightarrow BF} = \sum_j U_j^{A\rightarrow B} \ox \ket{f_j}^F\bra{e_j}^E,
\]
with qubit unitaries $U_j$ and with suitable orthonormal bases 
$\{\ket{e_j}\}$ and $\{\ket{f_j}\}$ of $E$ and $F$, respectively.

\medskip
To turn the other way, what are the useless unitary interactions, i.e.~those
with $Q_H(V)=0$, or at least $Q_{H\ox}(V)=0$? In the next section we will
encounter some families of two-qubit $V$ with the latter property. 
On the other hand, unitaries with $Q_H(V)=0$ do not seem to be so
obvious, except for the example of $\SWAP$, which swaps two isomorphic
systems $A$ and $E$, i.e.~$\SWAP(\ket{\psi}^A\ket{\varphi}^E) = \ket{\varphi}^B\ket{\psi}^F$,
because it results in channels with constant output.


\section{Two-qubit unitaries}
\label{sec:two-qubit}
In this section we will look at two-qubit unitary interactions, hence in principle 
study all qubit channels which can be described by a single qubit environment. 
This is motivated by quantum channels deriving from such 
unitaries having nice properties, which allow us to characterize their 
environment-assisted capacities.

A general two-qubit unitary interaction can be described by $15$ real parameters. 
For the analysis of quantum capacity under consideration we follow the arguments 
used in~\cite{KC01} to reduce the parameters to $3$ by the action of local 
unitaries. 

\medskip
\begin{lemma}[Kraus/Cirac~\cite{KC01}]
  Any two-qubit unitary interaction is equivalent, up to local unitaries before and
  after the gate, to one of the form 
  \[\begin{split}
    U^{AE} &= \sum_k e^{-i\lambda_{k}}\proj{\Phi_{k}}    \\
           &= \exp\bigl(-\alpha_x \sigma_x\ox\sigma_x 
                        -\alpha_y \sigma_y\ox\sigma_y 
                        -\alpha_z \sigma_z\ox\sigma_z\bigr),
  \end{split}\]
  with 
  \begin{align*}
    \lambda_{1} &= \frac{\alpha_{x} - \alpha_{y} + \alpha_{z}}{2}, \\
    \lambda_{2} &= \frac{-\alpha_{x} + \alpha_{y} + \alpha_{z}}{2}, \\
    \lambda_{3} &= \frac{-\alpha_{x} - \alpha_{y} - \alpha_{z}}{2}, \\
    \lambda_{4} &= \frac{\alpha_{x} + \alpha_{y} - \alpha_{z}}{2}, 
  \end{align*}
  and $\vert\Phi_k\rangle$ the so-called ``magic basis''~\cite{HW97},
  \begin{align*}
    \vert \Phi_{1} \rangle &= \frac{\vert 00 \rangle + \vert 11 \rangle}{\sqrt{2}}, \\
    \vert \Phi_{2} \rangle &= \frac{-i(\vert 00 \rangle - \vert 11 \rangle)}{\sqrt{2}}, \\
    \vert \Phi_{3} \rangle &= \frac{\vert 01 \rangle - \vert 10 \rangle}{\sqrt{2}}, \\
    \vert \Phi_{4} \rangle &= \frac{-i(\vert 01 \rangle + \vert 10 \rangle)}{\sqrt{2}}.
  \end{align*}
  This is of course the familiar Bell basis, but note the peculiar phases.
  \qed
\end{lemma}

\medskip
According to the definition of the capacities, the local unitaries on $A$, $B$, $E$ and $F$ 
do not affect the environment-assisted quantum capacity, as they could be incorporated 
into the encoding and decoding maps, respectively, or can be reflected in a different choice
of environment state.
The parameter space $(\alpha_{x},\alpha_{y},\alpha_{z})$ is further restricted by using
the following properties:
\begin{equation}
  U(\alpha_{x},\alpha_{y},\alpha_{z}) 
        = -i(\sigma_x \ox \sigma_x)\, U(\alpha_{x} + \pi,\alpha_{y},\alpha_{z}),
\end{equation}
and similarly 
\begin{equation}
  U\left(\frac{\pi}{2}\!+\!\alpha_{x},\alpha_{y},\alpha_{z}\right) 
    \!=\! 
  -i(\sigma_x \ox \1)\, U^*\left(\frac{\pi}{2}\!-\!\alpha_{x},\alpha_{y},\alpha_{z}\right) (\1 \ox \sigma_x),
\end{equation} 
where $U^*$ is the complex conjugate of $U$. 
Note that the latter has the same 
environment-assisted quantum capacities; indeed, any code for $U$ is transformed
into one for $U^*$ by taking complex conjugates.

Hence the parameter space given by
\begin{equation}
\label{Parspace}
  \mathfrak{T} = \left\{ (\alpha_x,\alpha_y,\alpha_z) : 
                         \frac{\pi}{2} \geq \alpha_x \geq \alpha_y \geq \alpha_z \geq 0 \right\}
\end{equation}
describes all two-qubit unitaries up to local basis choice and
complex conjugation.
This forms a tetrahedron with vertices $(0,0,0)$, $(\frac{\pi}{2},0,0)$,
$(\frac{\pi}{2},\frac{\pi}{2},0)$ and $(\frac{\pi}{2},\frac{\pi}{2},\frac{\pi}{2})$,
see Fig.~\ref{fig:uad}. 
Familiar two-qubit gates can easily be identified within this parameter space: 
for instance, $(0,0,0)$ represents the identity $\1$, 
$(\frac{\pi}{2},0,0)$ the CNOT, $(\frac{\pi}{2},\frac{\pi}{2},0)$ the DCNOT (double
controlled not), and $(\frac{\pi}{2},\frac{\pi}{2},\frac{\pi}{2})$ 
the SWAP gate, respectively.

\medskip
\begin{example}
To illustrate this parametrization, let us look at a controlled-unitary 
$V$ (cf.~Theorem~\ref{thuni}) of the form 
$V= \vert 0 \rangle \langle 0 \vert \ox U_0 + \vert1\rangle\langle1\vert \ox U_1$, 
where  $U_i \in {\rm SU}(2)$. 
One can work out that this has parametric representation $(t,0,0)$,
i.e.~in the parameter tetrahedron $\mathfrak{T}$, 
these unitaries are on the edge joining the identity $\1$ and CNOT. 

To see this, we use the argument described in Appendix A of~\cite{HVC02}: 
Observe that the spectrum of $V^T V$ is 
$(e^{-2i\lambda_1} ,e^{-2i\lambda_2},e^{-2i\lambda_3},e^{-2i\lambda_4} )$, 
where the transpose operator is with respect to the magic basis.
In this way, $(\vert 0\rangle \langle0\vert \ox \1)^{T} = \vert 1\rangle \langle1\vert \ox \1$
and $(\1 \ox U)^{T} = (\1 \ox U^{\dag})$, thus 
$V^{T} = \vert 1\rangle \langle 1\vert \ox U_0^{\dag} + \vert 0 \rangle \langle 0 \vert \ox U_1^{\dag}$ 
and 
$V^{T} V = \vert 1 \rangle \langle 1 \vert \ox U_0^{\dag}U_1 + \vert 0 \rangle \langle 0 \vert \ox U_1^{\dag} U_0$. 
The eigenvalues of  $U = U_0^{\dag}U_1$ are $e^{id}$ and $e^{-id}$, where $2\cos d = \Tr U$. 
The spectrum of thus $V^{T}V$ is $\left( e^{id},e^{id},e^{-id},e^{-id} \right)$. 
Using the order property 
$\frac{\pi}{2} \geq\lambda_4 \geq \lambda_1 \geq \lambda_2 \geq \lambda_3 \geq -\frac{3\pi}{4}$ 
(condition \eqref{Parspace} written in terms of $\lambda_k$) and solving the linear 
equations in $\alpha_x$, $\alpha_y$ and $\alpha_z$, we get the parametric point 
as $(t,0,0)$ where $t=d$ when $d \leq \frac{\pi}{2}$ and $t = \pi - d$ when $d \geq \frac{\pi}{2}$.
\end{example}

\medskip
Now we come to the main reason why we investigate this class of unitaries,
apart from obviously furnishing the smallest possible examples:
Recall that a quantum channel $\mathcal{N}:\mathcal{L}(A) \rightarrow \mathcal{L}(B)$ 
is called \emph{degradable}~\cite{DS05}
if there exists a degrading CPTP map $\mathcal{M}:\mathcal{L}(B) \rightarrow \mathcal{L}(F)$ such 
that for any input $\rho^{A}$, 
$\widetilde{\mathcal N}(\rho) = \mathcal{M}(\mathcal{N}(\rho))$. 
That is, Bob can simulate the environment output by applying a CPTP map on his 
system. It means that the complementary channel is noisier than the 
channel itself, in an operationally precise sense.

A quantum channel is \emph{anti-degradable} if its complementary channel is degradable, 
i.e.~if there exists a CPTP map $\mathcal{M}:\mathcal{L}(F) \rightarrow \mathcal{L}(B)$ such that for 
any input $\rho^{A}$, 
$\mathcal{N}(\rho^{AE}) = \mathcal{M}(\widetilde{\mathcal N}(\rho^{A}))$. 

It is well-known that the quantum capacity of anti-degradable channels 
is zero, by the familiar \emph{cloning argument}: Namely, if an anti-degradable 
channel were to have positive quantum capacity, $F$ can apply the degrading map
followed by the same decoder as $B$ and thus $A$ would be transmitting the
same quantum information to $B$ and $F$. 
This is in contradiction to the no-cloning theorem as observed in~\cite{BDS97}.
On the other hand, if a channel is degradable, Devetak and Shor~\cite{DS05}
showed that the quantum capacity can be characterized very concisely.
Namely, they proved that for degradable or anti-degradable 
$\mathcal{N}_i:\mathcal{L}(A_i) \longrightarrow \mathcal{L}(B_i)$,
\[\begin{split}
  \max_{\rho^{(n)}} I_c&\bigl(\rho^{(n)};\mathcal{N}_1\ox\cdots\mathcal{N}_n\bigr) \\
    &= \max_{\rho_1\ox\cdots\ox\rho_n} 
        I_c\bigl(\rho_1\ox\cdots\ox\rho_n;\mathcal{N}_1\ox\cdots\mathcal{N}_n\bigr) \\
    &= \sum_{i=1}^n \max_{\rho_i} I_c(\rho_i;\mathcal{N}_i),
\end{split}\]
which implies for degradable channel $\mathcal{N}$ that
\[
  Q(\mathcal{N}) = \max_\rho I_c(\rho;\mathcal{N}).
\]
Furthermore, the coherent information in this case is a concave function of 
$\rho$, so the maximum can be found efficiently.

Notice that by interchanging the registers in $B$ and $F$ we go from degradable 
channels to anti-degradable ones, and vice versa. But many channels are 
neither degradable nor anti-degradable. 
However, in~\cite{WPG07} it was shown that qubit channels with one qubit environment 
are either degradable or anti-degradable or both. Hence, for any initial state 
of the environment, all the two qubit unitary interactions give rise 
to qubit channels that are either degradable or anti-degradable or both. 
ref.~\cite{WPG07} also provided an analytical criterion for determining whether a channel 
is degradable or anti-degradable (or both, becoming \emph{symmetric} in such a case).  
The criterion is revisited here for our purposes.

\medskip
\begin{lemma}[Wolf/Perez-Garc\'{\i}a~\cite{WPG07}]
  \label{criterionDAD}
  Given an isometry $V : A \ox E \rightarrow B\ox F$ and an initial input to environment 
  $\ket{\eta}\in E$, let $\{K_i\}$ be the Kraus operators in normal form
  (i.e.~$\tr K_i^\dagger K_j = 0$ for $i\neq j$) of the qubit channel
  $\mathcal{N}_{\eta}(\rho) = \Tr_F V(\rho^A \ox \eta^E)V^\dag$.

  Then, the condition for degradability is given by the sign of the 
  $\det(2K_0^\dagger K_0 - \1)$. The channel is degradable when 
  $\det(2K_0^\dagger K_0 - \1) \geq 0$, anti-degradable when 
  $\det(2K_0^\dagger K_0 - \1) \leq 0$, and symmetric when $\det(2K_0^\dagger K_0 - \1) = 0$.
  \qed
\end{lemma}

\medskip
This characterization has the consequence that the separable environment-assisted
quantum capacity of two-qubit unitaries can be calculated fairly easily:

\medskip
\begin{theorem}
  \label{thm:2-qubit-Q-Htens}
  For a two-qubit unitary $V:AE \longrightarrow BF$,
  \[
    Q_{H\ox}(V) = \max_{\eta^E} \max_{\rho^A} I_c(\rho^A;\mathcal{N}_\eta).
  \]
  In addition, the maximization over helper states $\eta$ may be restricted
  to pure states such that $\mathcal{N}_\eta$ is degradable, and for each
  such fixed $\eta$, the inner maximization over $\rho$ is a convex
  optimization problem (concave function on a convex domain).
\end{theorem}
\begin{proof}
The capacity in general is given by Theorem~\ref{QH+QHtens}, Eq.~(\ref{eq:QHtens}):
\[
  Q_{H\ox}(V) = \sup_{n} \max_{\eta_1\ox\cdots\ox\eta_n} \max_{\rho^{(n)}}
       \frac{1}{n} I_c\bigl( \rho^{(n)};\mathcal{N}_{\eta_1}\ox\cdots\ox\mathcal{N}_{\eta_n} \bigr).
\]
By Wolf and Perez-Garcia's Lemma~\ref{criterionDAD}, each of the $\mathcal{N}_{\eta_i}$
is degradable or anti-degradable, so by Devetak and Shor~\cite{DS05}, the coherent
information is additive:
\[
  \max_{\rho^{(n)}} I_c\bigl(\rho^{(n)};\mathcal{N}_{\eta_1}\ox\cdots\mathcal{N}_{\eta_n}\bigr)
    = \sum_{i=1}^n \max_{\rho_i} I_c(\rho_i;\mathcal{N}_{\eta_i}),
\]
hence $Q_{H\ox}(V) = \max_\eta \max_\rho I_c(\rho;\mathcal{N}_\eta)$ as 
advertised.

Clearly, for those $\eta$ such that $\mathcal{N}_\eta$ is anti-degradable,
we know that the r.h.s.~is $0$, so we may discount them in the optimization.
\end{proof}

\medskip
\begin{definition}
  We say that a unitary operator $U$ to be \emph{universally degradable} 
  (resp. \emph{anti-degradable}), if for every $\ket{\eta} \in E$, 
  the qubit channel $\mathcal{N}_{\eta}:\mathcal{L}(A) \rightarrow \mathcal{L}(B)$ 
  is degradable (resp. anti-degradable).
  The set of universally degradable (anti-degradable) unitaries 
  is denoted $\mathfrak{D}$ ($\mathfrak{A}$).
\end{definition}

\medskip
Clearly, $\SWAP \in \mathfrak{A}$ and $\id\in\mathfrak{D}$, 
hence both $\mathfrak{A}$ and $\mathfrak{D}$ are non-empty.
Furthermore, $U\in\mathfrak{D}$ if and only if $\SWAP\cdot U\in\mathfrak{A}$.
Indeed, the set 
$\{(\alpha_x,\alpha_y,\alpha_z) \in \mathfrak{T}: U(\alpha_x,\alpha_y,\alpha_z) \in \mathfrak{A}\}$ 
is a tetrahedron with vertices $(\frac{\pi}{4},\frac{\pi}{4},\frac{\pi}{4})$,
$(\frac{\pi}{2},\frac{\pi}{4}$, $\frac{\pi}{4}),(\frac{\pi}{2},\frac{\pi}{2},0)$
and $(\frac{\pi}{2},\frac{\pi}{2},\frac{\pi}{2})$, shown in Fig.~\ref{fig:uad}.
Similarly, the set $\mathfrak{D}$ corresponds to the tetrahedron with vertices
$(0,0,0)$, $(\frac{\pi}{2},0,0)$, $(\frac{\pi}{4},\frac{\pi}{4},0)$ and
$(\frac{\pi}{4},\frac{\pi}{4},\frac{\pi}{4})$.
For a detailed analysis of the sets $\mathfrak{A}$ and $\mathfrak{D}$ and their
parameter regions we refer to Appendix~\ref{app:two-qubit}.

\begin{figure}[ht]
  \centering
  \includegraphics[width = 8cm, height = 6cm]{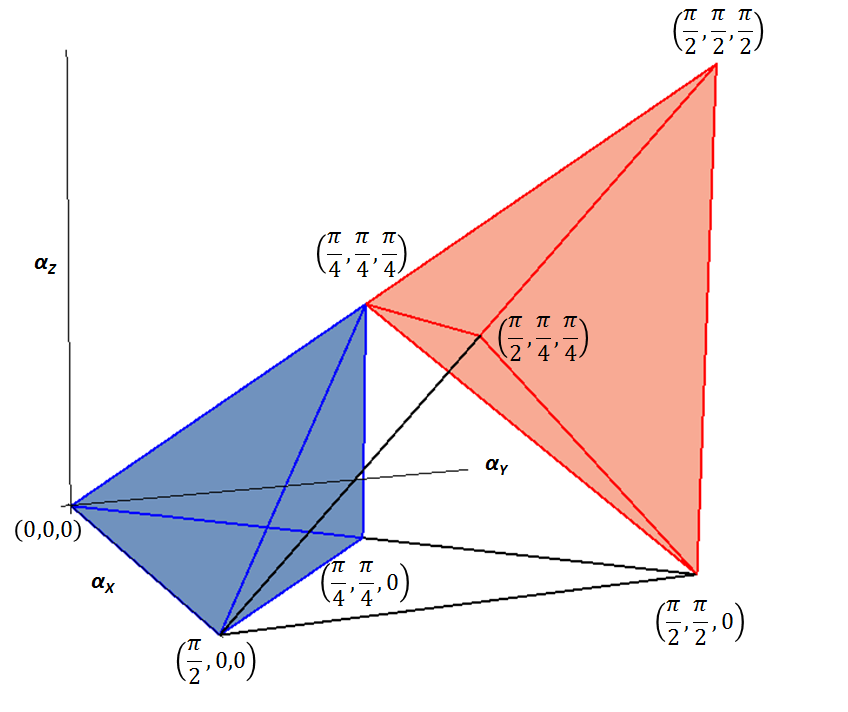}
  \caption{Universally anti-degradable and degradable regions inside the parameter 
     space $\mathfrak{T}$. The upper (red) tetrahedron corresponds to $\mathfrak{A}$,
     the lower (blue) one corresponds to $\mathfrak{D}$.}
  \label{fig:uad}
\end{figure}

Let us first consider the unique edge of the tetrahedron $\mathfrak{T}$
which contains points either belonging to $\mathfrak{A}$ or $\mathfrak{D}$.
This is the line segment joining the identity $\1$ $(0,0,0)$ with
SWAP $(\frac{\pi}{2},\frac{\pi}{2},\frac{\pi}{2})$.  
Each unitary on that line is a $\gamma$-th root of 
SWAP with a parameter $\gamma \in (0,1)$, i.e. 
\begin{align}
  \SWAP^{\gamma} &=      \frac{1 + e^{i\pi \gamma}}{2}\1 + \frac{1 - e^{i\pi \gamma}}{2}\SWAP \nonumber\\
                 &\equiv U\left( \frac{\gamma\pi}{2}, \frac{\gamma\pi}{2},\frac{\gamma\pi}{2}\right).
  \label{eq:gamma}
\end{align}

\medskip
It is actually elementary to evaluate the universally anti-degradable region 
of this line segment. 
Due to the invariance of $\SWAP$ under conjugation with unitaries of the form $u\ox u$,
it is enough to examine the anti-degradability of the channel that arise when the 
initial state of the environment is $\ket{0}$: either all $\mathcal{N}_\eta$ are
anti-degradable or none.
The Kraus operators are 
\begin{equation*}
\left[\begin{array}{cc}
1 & 0  \\
0 & \frac{1 + e^{i\pi \gamma}}{2} \\
\end{array}\right],\quad
\left[\begin{array}{cc}
0 & \frac{1 - e^{i\pi \gamma}}{2} \\
0 & 0 \\
\end{array}\right],
\end{equation*}
making it a generalized amplitude damping 
channel with damping parameter $\frac{1+e^{i\pi\gamma}}{2}$. 

\medskip
Hence we can invoke the criterion of Lemma~\ref{criterionDAD}, 
as these Kraus operators are in normal form.
It results that $\mathcal{N}_{\proj{0}}$ is anti-degradable for $\gamma \in [\frac12,1]$,
i.e.~$U\left( \frac{\gamma\pi}{2}, \frac{\gamma\pi}{2},\frac{\gamma\pi}{2}\right) \in \mathfrak{A}$.

From the above arguments it follows that
$Q_{H\ox}\bigl(U\left( \frac{\gamma\pi}{2}, \frac{\gamma\pi}{2},\frac{\gamma\pi}{2}\right)\bigr) = 0$ 
for $\gamma \in [\frac12,1]$. We do not know whether it is even true that
$Q_H(\SWAP^\gamma) = 0$ for these values of $\gamma$, which would require to show
that $Q_{H\ox}\bigl( (\SWAP^\gamma)^{\ox n} \bigr) = 0$ for all integers $n$.


\section{Super-Activation}
\label{sec:CA}
The significance of $U \in \mathfrak{A}$ is that a Helen restricted to $n$-separable
environment states cannot help Alice to communicate quantum information to Bob,
$Q_{H \ox}(U) = 0$, in accordance with Theorem~\ref{thm:2-qubit-Q-Htens}. 
The natural question now arising is whether an unrestricted Helen can perform any better. 
In this section we show that this can indeed be the case.

\subsection{Two different unitaries}
\label{subsec:SA}
The edges of the universally 
anti-degradable tetrahedron (Fig.~\ref{fig:uad}) 
provide examples of super-activation
($Q_{H\ox}(W) = Q_{H\ox}(V) = 0$ and $Q_{H\ox}(W \ox V) > 0$). 
These are discussed below by referring to the setting and notation of  
Fig.~\ref{fig:super}. 
The input state we will consider below in all the further analysis, unless mentioned otherwise,
shall be
$\vert 0 \rangle^{A'} \otimes \vert\Phi\rangle^{E' E} \otimes \vert\Phi\rangle^{A R}$,
where $\vert\Phi\rangle$ is the two-qubit maximally entangled state.

The global unitary $G$ is given by $W \ox V \ox \1_R$, so that
the coherent information is given by $S(\rho^{B' B}) - S(\rho^{F' F})$,
where $\rho^{B' B} = \Tr_{F' F R} G\proj{\Psi} G^{\dag}$ and 
$\rho^{F' F} = \Tr_{B' B R} G\proj{\Psi} G^{\dag}$ are the output states 
of Bob and Eve, respectively.

\begin{figure}[ht]
\centering
\begin{tikzpicture}[scale=0.4]
\draw[thick, green] (0,0) -- (1,0);
\draw[thick, green](1,0) -- (3,2) -- (7,2);
\draw[thick, green](1,0) -- (3,-2) -- (7,-2);
\draw[thick, red](3,4)--(7,4); \draw[thick,red](1,-6) -- (3,-4) -- (7,-4);
\draw[ultra thick] (7,1) rectangle (11,5); \draw[ultra thick] (7,-5) rectangle (11,-1);
\draw[thick, blue] (11,4) -- (17,4);
\draw[thick, blue] (11,-4) -- (17,-4);
\draw[thick, purple](11,2) -- (14,2); \draw[thick, purple](11,-2) -- (14,-2);
\draw[thick,purple](1,-6) -- (3,-8) -- (14,-8);
\node[cloud, cloud puffs=15.7, cloud ignores aspect, minimum width=0.7cm, minimum height=2cm, align=center, draw, fill=purple!50] (cloud) at (14,0){}; 
\node[cloud, cloud puffs=15.7, cloud ignores aspect, minimum width=0.3cm, minimum height=1cm, align=center, draw, fill=purple!50] (cloud) at (14,-8){} ; 
\draw[fill = blue!50] (17,0) ellipse (0.5 and 4.4);
\node[left] at (0,0){$H$};
\node[above] at (5,4){$A^{'}$}; \node[below] at (5,-4){$A$}; \node[above] at (9,-8){$R$};
\node[below] at (5,2){$E^{'}$}; \node[above] at (5,-2){$E$};
\node[below] at (12.8,2){$F^{'}$}; \node[above] at (12.8,-2){$F$};
\node[above] at (14,4){$B^{'}$}; \node[below] at (14,-4){$B$};
\draw (9,3) node[font = \fontsize{40}{42}\sffamily\bfseries]{$W$};
\draw (9,-3) node[font = \fontsize{40}{42}\sffamily\bfseries]{$V$};
\end{tikzpicture}
  \caption{The inputs controlled by Alice are  $A'$ and $A$, $R$ is the purification of 
           $A$. Helen controls $E'$ and $E$, Bob's systems are labelled as $B'$ and $B$. 
           The inaccessible output-environment systems are labelled as $F'$ and $F$. 
           Alice inputs $\vert 0 \rangle$ in $A'$ and $\ket{\Phi}$ in $AR$. Helen inputs 
           a Bell state $\ket{\Phi}$ in $E' E$.}
  \label{fig:super}
\end{figure}
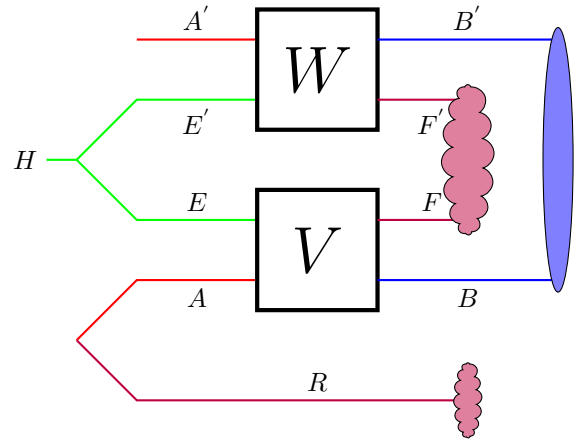

\begin{enumerate}[label={\emph{A-\arabic{*}}}]
\item Let $W$ be a unitary on the edge joining $\SWAP$ and DCNOT, 
  i.e.~$W = U(\frac{\pi}{2},\frac{\pi}{2},\frac{t\pi}{2})$ with a parameter 
  $t \in [0, 1]$; $V = \SWAP^{\gamma}$ with $\gamma \in [0.5,1]$.
  Then $W$ has $\lambda_{1} = \frac{t\pi}{4}$, 
  $\lambda_{2}= \frac{t\pi}{4}$, $\lambda_{3} = -\frac{\pi}{4}(t + 2)$
  and $\lambda_{4} = -\frac{\pi}{4}(t -2)$. 
  Hence, $W = e^{\frac{it\pi}{4}} \tilde{U}$ where
  \begin{equation*}
  \tilde{U} = 
  \left[\begin{array}{cccc}
    e^{-\frac{it\pi}{2}} & 0 & 0 & 0 \\
    0 &  0 & -i & 0 \\
    0 & -i & 0 & 0 \\
    0 & 0 & 0 & e^{-\frac{it\pi}{2}}
  \end{array}\right],
  \end{equation*}
  written in the computational basis.
  Bob's output state is then given by
  \[\begin{split}
    \rho^{B' B}
      &= \frac{1}{4}\left[ \frac{3 - \cos \pi\gamma}{2} \left(\proj{00} + \proj{11} \right) \right.   \\
      &\phantom{===}
       + i e^{-\frac{it\pi}{2}} (1- \cos \pi\gamma) \ket{00}\!\bra{11} \\
      &\phantom{===}
       - i e^{\frac{it\pi}{2}} (1- \cos \pi\gamma) \ket{11}\!\bra{00} \\
      &\phantom{===}
       + \left. \frac{1 + \cos \pi\gamma}{2}\bigl( \ket{01}\!\bra{01} + \ket{10}\!\bra{10} \bigr)\right],
  \end{split}\] 
  whose eigenvalues are
  $\frac{5-3\cos \pi\gamma}{8}$ (single) and $\frac{1+ \cos \pi\gamma}{8}$ (triple),
  while $\rho^{F' F} = \proj{0}^{F'} \ox \frac{1}{2} \1^F$.
  The coherent information vanishes at $\gamma^* \approx 0.6649$, see Fig.~\ref{fig:act1}. 
  Hence each unitary $U(\frac{\pi}{2},\frac{\pi}{2},\frac{t\pi}{2})$ with $t \in [0, 1]$ 
  super-activates $\SWAP^{\gamma}$ for $\gamma \in [0.5,\gamma^*)$.

  \begin{figure}[ht]
    \centering
    \includegraphics[height = 5cm, width = 8cm]{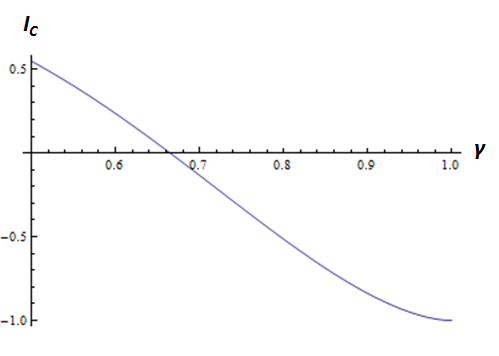}
    \caption{Example A-1: Plot of the coherent information $I_c = S(B' B) - S(F' F)$ when 
      $W=U(\frac{\pi}{2},\frac{\pi}{2},\frac{t\pi}{2})$ and
      $V= \SWAP^{\gamma}$, over $\gamma \in [0.5,1]$.}
    \label{fig:act1}
  \end{figure} 

\item Let $W = \SWAP$ and 
  $V = (\frac{\pi}{4} + \frac{t \pi}{4},\frac{\pi}{4},\frac{\pi}{4})$ with $t \in [0,1]$. 
  Here, $V$ sits on the edge joining $\sqrt{\SWAP}$ to $U(\frac{\pi}{2},\frac{\pi}{4},\frac{\pi}{4})$.
  The coherent information is positive for $t \in [0,1]$ as depicted in Fig.~\ref{fig:act2}.

  \begin{figure}[ht]
    \centering
    \includegraphics[height = 5cm, width = 8cm]{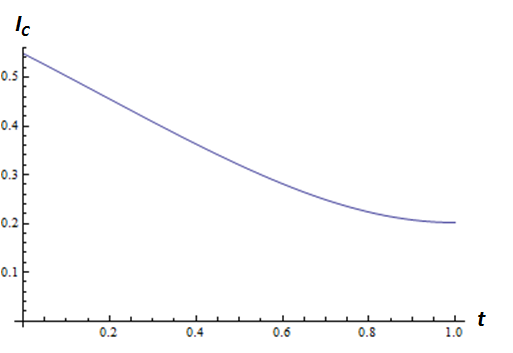}
    \caption{Example A-2: Plot of the coherent information $I_c = S(B' B) - S(F' F)$ when $W=\SWAP$ and 
             $V= (\frac{\pi}{4} + \frac{t \pi}{4},\frac{\pi}{4},\frac{\pi}{4})$, over $t \in [0,1]$.}
    \label{fig:act2}
  \end{figure} 

\item $\sqrt{\SWAP}$ activates $U(\frac{\pi}{2},\frac{\pi}{2},\frac{t\pi}{2})$ 
  for $t \in [0,1]$ as shown in example A-1. The coherent information is given by the 
  curve $s$ in Fig.~\ref{fig:act3}. Here let us evaluate the coherent information 
  for the setting described in Fig.~\ref{fig:super}, when we have 
  $V = \sqrt{\SWAP}$ and $W$ is a unitary on the edges of the tetrahedron corresponding
  to $\mathfrak{A}$. 
  By varying the parameter $t$ from $[0,1]$ we move along one of the edges of $\mathfrak{A}$. 
  \begin{enumerate}
    \item The edge joining $\sqrt{SWAP}$ to DCNOT:  
      $W = U(\frac{\pi}{4}+\frac{t\pi}{4},\frac{\pi}{4}+\frac{t\pi}{4},\frac{\pi}{4}-\frac{t\pi}{4})$. 
      The coherent information is given by the curve $p$ in Fig.~\ref{fig:act3}, 
      which is positive for $t \in (0,1]$.
    \item The edge joining $\sqrt{\SWAP}$ to $\SWAP$ ($\SWAP^\gamma$): 
      $W = U(\frac{\pi}{4} + \frac{t \pi}{4},\frac{\pi}{4} +  \frac{t \pi}{4},\frac{\pi}{4} + \frac{t \pi}{4})$.
      The coherent information is given by the curve $p$ in Fig.~\ref{fig:act3}, 
      which is positive for $t \in (0,1]$. 
      Here $t = 2\gamma - 1$, and the coherent information is 
      positive for $\gamma \in (\frac{1}{2},1]$.
    \item The edge joining $U(\frac{\pi}{2} ,\frac{\pi}{4},\frac{\pi}{4} )$ to $\SWAP$:  
      $W = U(\frac{\pi}{2} ,\frac{\pi}{4} +  \frac{t \pi}{4},\frac{\pi}{4} + \frac{t \pi}{4})$.
      The coherent information is given by the curve $q$ in Fig.~\ref{fig:act3}, which 
      is positive for $t \in [0,1]$.
    \item The edge joining $U(\frac{\pi}{2} ,\frac{\pi}{4},\frac{\pi}{4} )$ to DCNOT: 
      $W = U(\frac{\pi}{2} ,\frac{\pi}{4} +  \frac{t \pi}{4},\frac{\pi}{4} - \frac{t \pi}{4})$. 
      The coherent information is given by the curve $q$ in Fig.~\ref{fig:act3}, 
      which is positive for $t \in [0,1]$.
    \item The edge joining $\sqrt{\SWAP}$ to $U(\frac{\pi}{2} ,\frac{\pi}{4},\frac{\pi}{4} )$: 
      $W = U(\frac{\pi}{4} + \frac{t \pi}{4},\frac{\pi}{4} ,  \frac{\pi}{4})$. 
      The coherent information is given by the curve $r$ in Fig.~\ref{fig:act3},
      which is positive for $t \in (0,1]$.
  \end{enumerate}

  It results that each unitary corresponding to a point on the edge of the
  tetrahedron $\mathfrak{A}$ is super-activated by some 
  another $V \in \mathfrak{A}$. Actually a single unitary, $V = \sqrt{\SWAP}$, super-activates  
  every other unitary on the edges of the universally anti-degradable tetrahedron 
  (except itself).
  Furthermore, from the numerical analysis we have that $V = \sqrt{\SWAP}$ super-activates 
  every $W \in \mathfrak{A}$ (except itself). 

  \begin{figure}[ht]
    \centering
    \includegraphics[height = 5cm, width = 8cm]{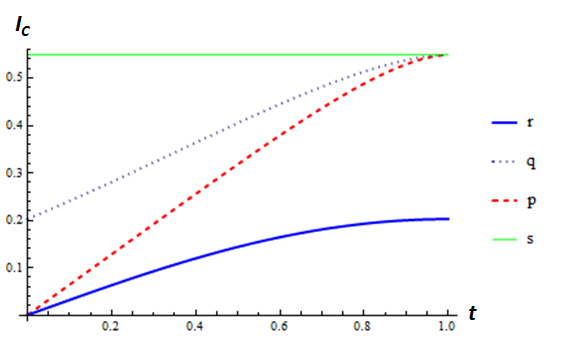}
    \caption{Plots of the coherent information when $V=\sqrt{\SWAP}$ and $W$ is on one of the edges 
    of the tetrahedron $\mathfrak{A}$, examples A-3a through A-3e.}
    \label{fig:act3}
  \end{figure} 
\end{enumerate}

Thus, in all the above cases,
\[
  Q_{H\ox}(W \otimes V) > Q_{H\ox}(W) + Q_{H\ox}(V) = 0.
\]
In other words, two seemingly useless unitaries can transfer a positive rate of
quantum information when used in conjunction and the input environments are entangled. 
All the above $W$ and $V$ show superactivation of $Q_{H\ox}$. In addition,
in the examples A-1 and A-2, we have $W=\SWAP$, hence in fact even $Q_H(W)=0$.
In particular the roots $\SWAP^\gamma$ of the SWAP gate are interesting. When  
$\sqrt{\SWAP}$ is used in conjunction with a different $W \in \mathfrak{A}$ and 
the input environments are entangled, then they could transfer positive quantum 
information i.e. $Q_{H\ox}(\sqrt{\SWAP} \ox W) > 0$.


\subsection{Self-super-activation}
\label{subsec:SSA}
So far we have considered two different unitaries. The question is if two copies of 
the same unitary ($\in \mathfrak{A}$) can yield positive capacity
when the initial states environments are entangled? In other words,
can $Q_{H\ox}$ be \emph{self-super-activated}?
The answer to this question is affirmative as we shall show now. 

\medskip
\begin{remark}
  From the super-activation of a unitary $W$ with another unitary $V$,
  such that both $W$ and $V$ are universally anti-degradable, we can 
  get a self-super-activating unitary by doubling the size of the environment. 
  More precisely, we can construct the new unitary 
  $R : A\ox E\ox{E'} \rightarrow B\ox F\ox {F'}$, with $E'=F'=\CC^2$: 
  $R = W^{AE}\ox\proj{0}^{E'} + V^{AE}\ox\proj{1}^{E'}$.
  
  To see that this works, clearly if Helen inputs $\ket{0}$ into $E'$,
  she determines that the unitary on $AE$ is $W$, if she inputs $\ket{1}$
  into $E'$, the unitary is $V$; hence from two uses, $R\ox R$, she can 
  get $W\ox V$, which has positive environment-assisted capacity by
  assumption. On the other hand $Q_{H\ox}(R)=0$, because in fact $R$
  is itself universally anti-degradable.
  Namely, observe that if the channels induced by $W$ and $V$ for
  environment input states $\psi$ and $\varphi$ are denoted by
  $\mathcal{N}_\psi$ and $\mathcal{M}_\varphi$, respectively, then a generic input
  state $\sqrt{p}\ket{\psi}\ket{0} + \sqrt{1-p}\ket{\varphi}\ket{1}$
  to the $EE'$ registers of $W$ results in the channel
  $p\mathcal{N}_\psi + (1-p)\mathcal{M}_\varphi$. As both components are
  anti-degradable, so is their convex combination.
\end{remark}

\medskip
However, by looking at our two-qubit classification more carefully, we can also
find self-super-activation in this simplest possible setting.

\begin{enumerate}[label={\emph{B-\arabic{*}}}]
\item Let us consider the unitaries
  $U\left(\frac{\pi}{4} + t \frac{\pi}{4},\frac{\pi}{4} 
          + t \frac{\pi}{4},\frac{\pi}{4} - t \frac{\pi}{4}\right)$  with $t \in [0,1]$. 
  We have seen in example A-3a that these unitaries are activated by 
  $\sqrt{\SWAP}$ in $t \in (0,1]$. Now we shall explore the case when 
  $W = V = U\left(\frac{\pi}{4} + t \frac{\pi}{4},\frac{\pi}{4} 
          + t \frac{\pi}{4},\frac{\pi}{4} - t \frac{\pi}{4}\right)$.
  The coherent information $S(BB') - S(FF')$ is positive for $t \in (0,1)$ 
  as shown by curve $m$ in Fig.~\ref{fig:self1}.
  
  \begin{figure}[ht]
    \centering
    \includegraphics[height = 5cm, width = 8cm]{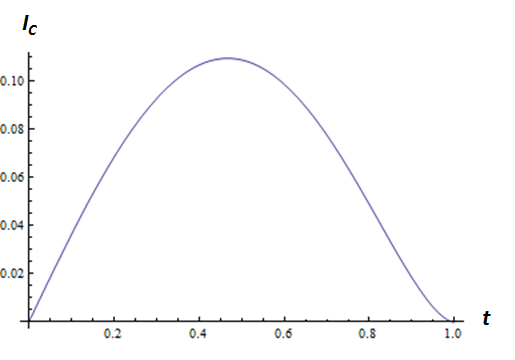}
    \caption{Example B-1: Plot of the coherent information for the family
     $W = V = 
      U\left(\frac{\pi}{4} + \frac{t \pi}{4},\frac{\pi}{4} + \frac{t \pi}{4},\frac{\pi}{4} - \frac{t \pi}{4}\right)$; 
     it is positive for $t \in (0,1)$}
    \label{fig:self1}
  \end{figure} 
  
  When Helen can create quantum correlation between the environment inputs we 
  see that a seemingly ``useless'' unitary can transmit quantum information. That 
  is, the unrestricted Helen can super-activate the interaction
  $U\left(\frac{\pi}{4} + t \frac{\pi}{4},\frac{\pi}{4} 
          + t \frac{\pi}{4},\frac{\pi}{4} - t \frac{\pi}{4}\right)$, 
  with $t \in (0,1)$ which translates to 
  \begin{align}
    Q_H&\left(U\left(\frac{\pi}{4} + t \frac{\pi}{4},\frac{\pi}{4} 
                     + t \frac{\pi}{4},\frac{\pi}{4} - t \frac{\pi}{4}\right)\right) \nonumber\\
       &\!\!\!\!\!\!\!\!
        > Q_{H \ox}\left(U\left(\frac{\pi}{4} + t \frac{\pi}{4},\frac{\pi}{4} 
                        + t \frac{\pi}{4},\frac{\pi}{4} - t \frac{\pi}{4}\right)\right) = 0.
  \end{align}

\item We can provide another family of unitaries which exhibit self-super-activation by the
  unitaries
  $W=V = U\left(\frac{\pi}{2} ,\frac{\pi}{4} + t \frac{\pi}{4},\frac{\pi}{4} - t \frac{\pi}{4}\right)$.
  The environment input state is 
  $\vert \Psi \rangle_{\theta} 
   = \ket{1}^{A'} \ox \ket{\Phi}^{E'E} \ox \left(\sqrt{\theta}\ket{00} + \sqrt{1-\theta}\ket{11}\right)^{AR}$.
  By optimizing over $\theta$, we numerically find positive coherent information for
  $t \in (0.0004,0.9999)$. The plots in in Fig.~\ref{fig:self2} show $\theta = \frac12$ (curve $m$),
  $\theta = 2^{-6}$ (curve $n$) and $\theta=2^{-10}$ (curve $o$).
  The coherent information achievable seems to get smaller and smaller as $t$ 
  approaches $0$.
%
  
  \begin{figure}[ht]
    \centering
    \includegraphics[height = 5cm, width = 8cm]{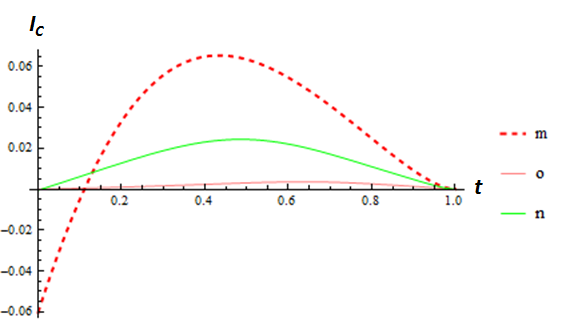}
      \caption{Example B-2: Plots of the coherent information for the family
       $W=V = U\left(\frac{\pi}{2} ,\frac{\pi}{4} + t \frac{\pi}{4},\frac{\pi}{4} - t \frac{\pi}{4}\right)$.
       Curves $m$, $n$, $o$ correspond to input states $\vert \Psi\rangle_\theta$ with
       $\theta = \frac12,2^{-6},2^{-10}$, respectively.}
    \label{fig:self2}
  \end{figure} 
\end{enumerate}

For all the above $U$, $Q_{H\ox}(U)=0$ but $Q_H(U)>0$, showing that
to unlock the full potential of an interaction $U$, the helper may
need to entangle the environments of different instances of $U$.

\medskip
\begin{remark}
  The phenomenon of self-super-activation taking place thanks to entanglement 
  across environments resembles the super-additivity of the capacity
  in quantum channels with memory~\cite{DBF07,LPM09}.
\end{remark}


\section{Entanglement-assisted helper}
\label{sec:entangled}
Entanglement played a pivotal role in the instances of superactivation 
exhibited above;
when Helen could create correlation between the environment input 
registers, she could enhance quantum communications from Alice to Bob. 
In this section 
we consider the model when there is pre-shared entanglement between 
Helen and Bob. This model is motivated by the equivalence of the two
schemes presented in Fig.~\ref{fig:swap}.

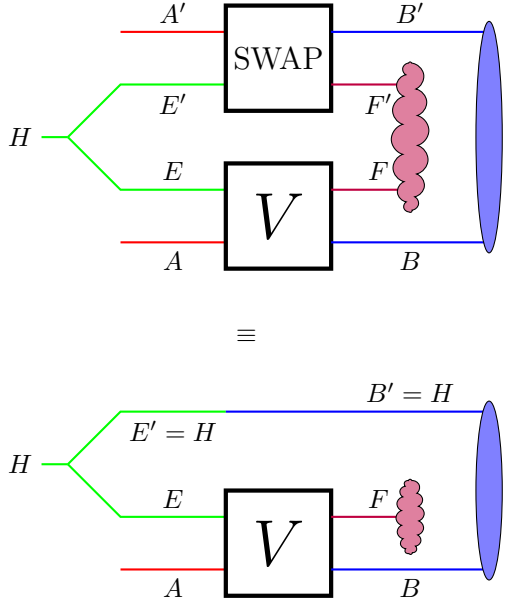
\begin{figure}[ht]
\centering
\begin{tikzpicture}[scale=0.35]
\draw[thick, green] (0,0) -- (1,0);
\draw [thick, green](1,0) -- (3,2) -- (7,2);
\draw [thick, green](1,0) -- (3,-2) -- (7,-2);
\draw [thick, red](3,4)--(7,4); \draw[thick,red] (3,-4) -- (7,-4);
\draw[ultra thick] (7,1) rectangle (11,5); \draw[ultra thick] (7,-5) rectangle (11,-1);
\draw[thick, blue] (11,4) -- (17,4);
\draw[thick, blue] (11,-4) -- (17,-4);
\draw [thick, purple](11,2) -- (14,2); \draw[thick, purple](11,-2) -- (14,-2);
\node[cloud, cloud puffs=15.7, cloud ignores aspect, minimum width=0.5cm, minimum height=2cm, align=center, draw, fill=purple!50] (cloud) at (14,0){} ; 
\draw[fill = blue!50] (17,0) ellipse (0.5 and 4.4);
\node[left] at (0,0){$H$};
\node[above] at (5,4){$A'$}; \node[below] at (5,-4){$A$};
\node[below] at (5,2){$E'$}; \node[above] at (5,-2){$E$};
\node[below] at (12.8,2){$F'$}; \node[above] at (12.8,-2){$F$};
\node[above] at (14,4){$B'$}; \node[below] at (14,-4){$B$};
\draw (9,3) node[font = \fontsize{13}{14}\sffamily\bfseries]{$\SWAP$};
\draw (9,-3) node[font = \fontsize{40}{42}\sffamily\bfseries]{$V$};
\end{tikzpicture}

\bigskip
\centering{$\equiv$}
\bigskip

\begin{tikzpicture}[scale=0.35]
\draw[thick, green] (0,0) -- (1,0);
\draw [thick, green](1,0) -- (3,2) -- (7,2);
\draw [thick, green](1,0) -- (3,-2) -- (7,-2);
 \draw[thick,red] (3,-4) -- (7,-4);
\draw[ultra thick] (7,-5) rectangle (11,-1);
\draw[thick, blue] (11,-4) -- (17,-4);\draw[thick,blue] (7,2) --(17,2);
 \draw[thick, purple](11,-2) -- (14,-2);
\node[cloud, cloud puffs=15.7, cloud ignores aspect, minimum width=0.3cm, minimum height=1cm, align=center, draw, fill=purple!50] (cloud) at (14,-2){} ; 
\draw[fill = blue!50] (17,-1) ellipse (0.5 and 3.4);
\node[left] at (0,0){$H$};
 \node[below] at (5,-4){$A$};
\node[below] at (5,2){$E'=H$}; \node[above] at (5,-2){$E$};
 \node[above] at (12.8,-2){$F$};
 \node[below] at (14,-4){$B$};\node[above] at (14,2){$B'=H$};
\draw (9,-3) node[font = \fontsize{40}{42}\sffamily\bfseries]{$V$};
\end{tikzpicture}
\caption{When inputting an entangled state across $E'E$ and an arbitrary state
  in $A'$ (top), the $\SWAP$ acts like a ``dummy'' channel but helps to establish 
  entanglement between the receiver $BB'$ and the environment $E$. 
  This is equivalent to sharing an entangled state between Helen and Bob (bottom).}
\label{fig:swap}
\end{figure}

$\SWAP$ merely exchanges the input and environment registers, which could 
be used to correlate the environment on the input side with the receiver 
when the initial environment states are entangled. Indeed, this was behind
several of the examples of super-activation in the previous section
(A-1 and A-2).

Extending the notation of $\mathcal{N}_\eta = \mathcal{N}(\cdot \ox \eta)$
introduced in Section~\ref{sec:main}, we let, for a state $\kappa$ on $EH$,
\[
  \mathcal{N}_\kappa^{A \rightarrow BH}(\rho) 
         := (\mathcal{N}^{AE \rightarrow B}\ox\id_H)(\rho^A\ox\kappa^{EH}).
\]

Referring to Fig.~\ref{fig:entast}, we can further define the following CPTP maps. 
An encoding map $\mathcal{E}:\mathcal{L}(A_{0}) \rightarrow \mathcal{L}(A^n)$, and 
the decoding map $\mathcal{D}:\mathcal{L}(B^n \ox H) \rightarrow \mathcal{L}(B_0)$. 
The output after the overall dynamics when we input a maximally entangled state
$\Phi^{RA_0}$, with the inaccessible reference system $R$, is given by
$\sigma^{RB_0} 
 =\mathcal{D}\left(\mathcal{N}^{\ox n} \ox \id_H 
                   \left(\mathcal{E}(\Phi ^{RA_0}) \ox \kappa^{E^n}\right) \right)$.

\medskip
\begin{definition}
  \label{def:entanglement-assisted-code}  
  An \emph{entangled environment-assisted quantum code} of block length $n$ is a 
  triple $(\mathcal{E}^{A_0 \rightarrow A^n},\kappa^{E^n H},\mathcal{D}^{B^n H \rightarrow B_0})$.
  Its \emph{fidelity} is given by $F = \tr \Phi^{RA_0}\sigma^{RB_0}$,
  and its \emph{rate} defined as $\frac{1}{n}\log|A_0|$. 
  
  A rate $R$ is called \emph{achievable} if there are codes of all block lengths
  $n$ with fidelity converging to $1$ and rate converging to $R$. The 
  \emph{entangled environment-assisted quantum capacity} of $V$, denoted
  $Q_{EH}(V)$, or equivalently $Q_{EH}(\mathcal{N})$, is the maximum achievable rate.
%
\end{definition}

\begin{figure}[ht]
\centering
\begin{tikzpicture}[scale=0.28]
\draw [ultra thick] (5,-11) rectangle (10,-1) ;
\draw [ultra thick] (18,-14) rectangle (23,-1) ;
\draw [ultra thick] (13,-3) rectangle (15,-1) node[midway]{$\mathcal{N}$};
\draw [ultra thick] (13,-6) rectangle (15,-4) node[midway]{$\mathcal{N}$};
\draw [ultra thick] (13,-11) rectangle (15,-9) node[midway]{$\mathcal{N}$};
\draw [thick, red] (0,0)   -- (3,-6) -- (5,-6);
\draw [thick, red] (10,-1.5) -- (13,-1.5);
\draw [thick, red] (10,-4.5) -- (13,-4.5);
\draw [thick, red] (10,-9.5) -- (13,-9.5);
\draw [thick,blue] (15,-2) -- (18,-2);
\draw [thick, blue] (15,-5) -- (18,-5);
\draw [thick, blue] (15,-10) -- (18,-10);
\draw [thick, blue] (23,-6) -- (25,-6) -- (28,0);
\draw [thick, purple] (0,0) -- (3,6) -- (25,6) -- (28,0);
\draw [loosely dashed] (-0.5,0) -- (28.5,0);
\draw [thick, green] (11,-13) -- (11,-2.5) -- (13,-2.5);
\draw [thick, green] (11,-13) -- (11.2,-5.5) -- (13,-5.5);
\draw [thick, green] (11,-13) -- (11.4,-10.5) -- (13,-10.5);
\node[above] at (4,-6) {$A_{0}$}; \node[above] at (24,-6){$B_{0}$};\node[below] at (14,6){$R$};
\node[left] at (0,0){$\Phi$};
\node[right] at (28,0){$\sigma$};
\node[above] at (10.5,-1.5){$A$};\node[above] at (10.5,-4.5){$A$};\node[above] at (10.5,-9.5){$A$};
\node[below] at (12.3,-2.5){$E$};\node[below] at (12.3,-5.5){$E$};\node[below] at (12.3,-10.5){$E$};\node[above] at (17,-2){$B$};\node[above] at (17,-5){$B$};\node[above] at (17,-10){$B$};
\node[below] at (11,-13){$\kappa$};
\draw (7.5,-6) node[font = \fontsize{40}{42}\sffamily\bfseries]{$\mathcal{E}$};
\draw (20.5,-7.5) node[font = \fontsize{40}{42}\sffamily\bfseries]{$\mathcal{D}$};
\draw [thick,dotted] (10.5,-6.5) -- (10.5,-7.5);
\draw[thick,dotted] (12.3,-7.5) -- (12.3,-8.5);
\draw[thick,dotted] (14,-7) -- (14,-8);
\draw[thick,dotted] (17,-7) -- (17,-8);
\draw[thick, blue] (11,-13) -- (18,-13);
\node[below] at (14.5,-13){$H$};
\end{tikzpicture}
\caption{The general form of a protocol to transmit quantum information when the helper 
         and the receiver pre-share entanglement; $\mathcal{E}$ and $\mathcal{D}$ are the encoding 
         and decoding maps respectively, $\kappa$ is the initial state of the environments and system $H$.}
\label{fig:entast}
\end{figure}
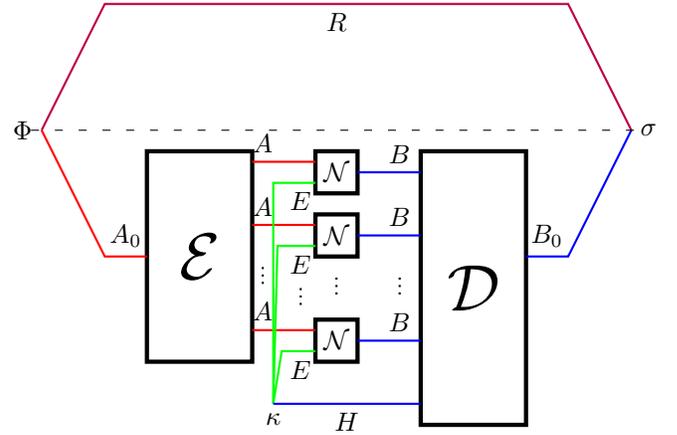

\medskip
\begin{theorem}
  \label{QEH}
  The entangled environment-assisted quantum capacity of an interaction 
  $V:AE \longrightarrow BF$ is characterized by following regularization.
  \begin{equation}\begin{split}
    \label{eq:QEH}
    Q_{EH}(V) &= \sup_n \max_{\ket{\kappa^{(n)}}\in E^nH} 
                  \frac1n Q\bigl( (\mathcal{N}^{\ox n})_{\kappa^{(n)}} \bigr) \\
              &= \sup_n \max_{\ket{\kappa}\in E^nH} \max_{\rho^{(n)}}
                  \frac1n I_c\bigl( \rho^{(n)};(\mathcal{N}^{\ox n})_{\kappa^{(n)}} \bigr).
  \end{split}\end{equation}
  The maximization is over (w.l.o.g.~pure) states $\kappa^{(n)}$ on $E^n H$ and
  input states $\rho^{(n)}$ on $A^n$.
\end{theorem}
\begin{proof}
  The direct part, i.e.~the ``$\geq$'' inequality, follows directly from the 
  LSD theorem~\cite{Lloyd96,Shor00,Devetak03}, applied to the channel
  $(\mathcal{N}^{\ox n})_{\kappa^{(n)}}$, to be precise asymptotically many
  copies of this block-channel, so that the i.i.d.~theorems apply~\cite{Wilde11}.

  The converse (``$\leq$''), works as before in Theorem~\ref{QH+QHtens},
  following Schumacher, Nielsen and Barnum~\cite{Schu96,NS96,BNS98}: 
  Consider a code
  of block length $n$ and fidelity $F$, where the helper uses an environment
  state $\kappa^{(n)}$; otherwise we use notation as in Fig.~\ref{fig:entast}.
  We have $\frac12 \|\sigma-\Phi\|_1 \leq \sqrt{1-F} =: \epsilon$,
  cf.~\cite{Fuchs-vandeGraaf}. Now, Fannes' inequality~\cite{Fannes} 
  can be applied, at least once $2\epsilon \leq \frac1e$ (i.e.~when $F$ is 
  large enough), yielding
  \[\begin{split}
    I(R\rangle B_0)_\sigma &=    S(\sigma^{B_0}) - S(\sigma^{RB_0}) \\
                           &\geq S(\sigma^{B_0}) \\
                           &\geq S(\Phi^{A_0}) - 2\epsilon \log|B_0| - H_2(2\epsilon) \\
                           &\geq (1-2\epsilon) \log|A_0| - 1.
  \end{split}\]
  On the other hand, with $\omega = (\id\ox\mathcal{E})\Phi$,
  \[\begin{split}
    I(R\rangle B_0)_\sigma 
             &\leq I(R\rangle B^n)_{(\id\ox\mathcal{N}^{\ox n}_{\kappa^{(n)}})\omega} \\
             &\leq \max_{\ket{\phi}^{RA^n}} I(R\rangle B^n)_{(\id\ox\mathcal{N}^{\ox n}_{\kappa^{(n)}})\phi} \\
             &=    \max_{\rho^{(n)}} I_c\bigl( \rho^{(n)};(\mathcal{N}^{\ox n})_{\kappa^{(n)}} \bigr),
  \end{split}\]
  using first data processing of the coherent information and then its
  convexity in the state~\cite{NS96}.
  As $n \rightarrow \infty$ and $F\rightarrow 1$, the upper bound on the rate
  follows.
\end{proof}

\medskip
\begin{proposition}
  The entangled environment-assisted quantum capacity is continuous.
  The statement and proof are analogous to the ones in 
  Proposition~\ref{prop:asymp-continuity}, following~\cite{LS09}.
  \qed
\end{proposition}

\medskip
The super-activation of $U \in \mathfrak{A}$ with $\SWAP$ depicted in 
Fig.~\ref{fig:super} translates to positive capacity of the entangled helper. 
We discuss two concrete examples of two-qubit unitaries:

\begin{enumerate}[label={\emph{E-\arabic{*}}}]
  \item $Q_{EH}(\SWAP^\gamma) > 0$ for $\gamma \in [0.5,0.6649)$,
        cf.~Section~\ref{sec:CA}, example A-1.

  \item Consider $U$ corresponding to a point on the line segment joining
        $(\frac{\pi}{4},\frac{\pi}{4},\frac{\pi}{4})$ and 
        $(\frac{\pi}{2},\frac{\pi}{4},\frac{\pi}{4})$. These points are vertices of 
        $\mathfrak{A}$ (see Fig.~\ref{fig:uad}), and hence the line segment is an 
        edge of the universally anti-degradable tetrahedron. As we saw in
        Section~\ref{sec:CA}, example A-2, this is super-activated by $\SWAP$.
\end{enumerate}

We now show how to evaluate the single-copy coherent information in the entangled 
environment-assisted capacity of $\SWAP^{\gamma}$, with $\gamma \in [0,1]$,
as per Theorem~\ref{QEH}, Eq.~(\ref{eq:QEH}); the setting is as in the lower 
part of Fig.~\ref{fig:swap}.
To proceed, we need the following lemma.

\medskip
\begin{lemma}
  \label{lemma:universal-degradability}
  If an isometry $U:AE \longrightarrow BF$ is universally degradable, then
  for every $\ket{\kappa} \in EH$, the channel 
  $\mathcal{N}_\kappa:\mathcal{L}(A) \longrightarrow \mathcal{L}(BH)$ is
  degradable.
\end{lemma}
\begin{proof}
Recall $\mathcal{N}_\kappa(\rho) = \tr_F (\mathcal{N}\ox\id_H)(\rho^A\ox\kappa^{EH})$,
with Stinespring dilation $V\ket{\varphi} = (U\ox\1)(\ket{\varphi}\ket{\kappa})$,
mapping $A$ to $BH \ox F$. Hence, the complementary channel is given by
\[
  \widetilde{\mathcal{N}_\kappa}(\rho) = \tr_B \mathcal{N}(\rho^A\ox\kappa^E),
\]
with the reduced state $\kappa^E = \tr_H \kappa$.

Let $\ket{\kappa} = \sum_i \sqrt{p_i}\ket{\eta_i}^E \ket{i}^H$ be the Schmidt
decomposition. Then, on the one hand,
\[
  \widetilde{\mathcal{N}_\kappa} = \sum_i p_i \widetilde{\mathcal{N}_{\eta_i}}
                                 = \sum_i p_i \mathcal{D}_i \circ \mathcal{N}_{\eta_i},
\]
with degrading CPTP maps $\mathcal{D}_i^{B\rightarrow F}$ by assumption.

As $i$ is accessible in the output of $\mathcal{N}_\kappa$ by measuring
$H$ in the computational basis, we obtain the degrading map 
$\mathcal{D}^{BH\rightarrow F}$ such that
$\widetilde{\mathcal{N}_\kappa} = \mathcal{D} \circ \mathcal{N}_\kappa$,
via
$\mathcal{D}(\sigma\ox\ket{i}\!\bra{j}) = \delta_{ij}\mathcal{D}_i(\sigma)$.
\end{proof}

\medskip
Returning to $\SWAP^\gamma$, the combined channel and environment input is
$\rho^A \ox \kappa^{EH}$. Because of the $u\ox u$-symmetry of 
the gate, we may without loss of generality choose the bases of $E$ and $H$ 
such that 
$\ket{\kappa}^{EH} = \sqrt{\lambda}\ket{00}+\sqrt{1-\lambda}\ket{11}$.

Now, $\kappa$ is invariant under the action of $Z^E \ox {Z^\dagger}^H$,
hence we obtain a covariance property of the channel:
\[
  \mathcal{N}_\kappa(Z\rho Z^\dagger) = (Z \ox Z^\dagger) \mathcal{N}_\kappa(\rho) (Z^\dagger \ox Z).
\]
By Lemma~\ref{lemma:universal-degradability}, $\mathcal{N}_\kappa$ is degradable,
hence the coherent information is concave in $\rho^A$~\cite{DS05}
and so the coherent information is maximized on an input density
$\rho^A$ that commutes with $Z$. I.e.~we may assume that 
$\rho^A = \mu\proj{0} + (1-\mu)\proj{1}$.

We then find for the output states of Bob ($B'B=HB$) and the environment ($F$) that
\[\begin{split}
 \rho^{B' B}
  &= \lambda \left( \mu +  (1-\mu) \left|\frac{1-e^{i\pi \gamma}}{2}\right|^{2} \right)\proj{00}    \\
  &\phantom{=}
   + (1-\lambda) \left( (1-\mu) + \mu \left|\frac{1-e^{i\pi \gamma}}{2}\right|^{2} \right)\proj{11} \\ 
  &\phantom{=}
   + \sqrt{\lambda(1-\lambda)} \left( \frac12 - \frac{\mu}{2} e^{-i \pi \gamma} 
                                              - \frac{1-\mu}{2} e^{i \pi \gamma} \right) \ket{00}\!\bra{11} \\
  &\phantom{=}
   + \sqrt{\lambda(1-\lambda)} \left( \frac12 - \frac{\mu}{2} e^{i \pi \gamma} 
                                              - \frac{1-\mu}{2} e^{-i \pi \gamma} \right) \ket{11}\!\bra{00} \\
  &\phantom{=}
   + \lambda(1-\mu) \left|\frac{1+e^{i\pi \gamma}}{2}\right|^{2} 
                      \bigl( \ket{01}\!\bra{01} + \ket{10}\!\bra{10} \bigr),
\end{split}\]
and $\rho^{F}$ is diagonal in the computational basis:
\[\begin{split}
  \rho^F \! &= \left( \! \lambda \mu \!+\! \lambda(1-\mu) \left|\frac{1+e^{i\pi \gamma}}{2}\right|^2 \!\!
                           \!+\! \mu(1-\lambda) \left|\frac{1-e^{i\pi \gamma}}{2}\right|^2 \right) \proj{0} \\
         &\phantom{==}
          + \left( (1-\lambda)(1-\mu) + \lambda(1-\mu) \left|\frac{1-e^{i\pi \gamma}}{2}\right|^2 \right. \\
         &\phantom{===============}
            \left.               + \mu(1-\lambda) \left|\frac{1+e^{i\pi \gamma}}{2}\right|^2 \right) \proj{1}.
\end{split}\]

\begin{figure}[ht]
  \centering
  \includegraphics[height = 5cm, width = 8cm]{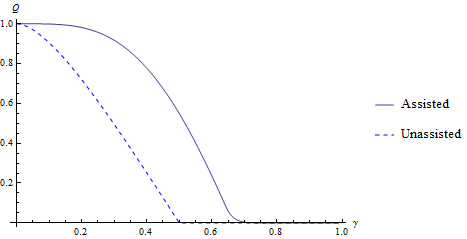}
  \caption{The unbroken curve in the plot is $Q_{EH\ox}$, the single-copy coherent information
    in the formula for the entangled environment-assisted 
    quantum capacity of $\SWAP^{\gamma}$, Eq.~(\ref{eq:QEH}),
    i.e.~the maximum of $I_c(\rho;\mathcal{N}_\kappa)$ over states 
    $\rho^A$ and $\kappa^{EH}$.
    The dashed line is the restricted environment-assisted
    quantum capacity $Q_{H\ox}(\SWAP^\gamma)$.}
  \label{fig:easwapgamma}
\end{figure}

In Fig.~\ref{fig:easwapgamma} we plot the single-copy coherent information
assisted by an entangled environment, maximized over $\lambda$ and $\mu$, 
and compare it with the same quantity without pre-shared entanglement. 
This is actually the quantum capacity assisted by entangled states
of the form 
$\kappa^{E^nH^n} = \kappa^{E_1H_1}\ox\cdots\ox\kappa^{E_nH_n}$ in
Definition~\ref{def:entanglement-assisted-code}, which we might
denote $Q_{EH\ox}(U)$ in analogy with $Q_{H\ox}(U)$.
As shown in the plot, the entanglement between Helen and Bob increases the 
quantum capacity of $\SWAP^{\gamma}$ to a positive quantity
for a large interval of $\gamma$ values, up to $\gamma^{**} \approx 0.7662$. 

\medskip
\begin{remark}
  It follows that we could achieve super-activations of $\SWAP^\gamma$ with $\SWAP$ 
  for larger interval of $\gamma \in [0.5, 0.7662)$, when optimizing  
  over the input of $\SWAP^\gamma$ and the initial environment state, in Section~\ref{sec:CA}, example A-1.
\end{remark}

\medskip
\begin{remark}
  We could even contemplate a fully entanglement-assisted model, where 
  both Alice and Helen share prior entanglement with Bob. This is a special 
  case of Hsieh \emph{et al.}'s entanglement-assisted multi-access
  channel~\cite{HDW08}: Indeed, if the
  achievable rate region of pairs of rates $(R_A,R_E)$ for quantum
  communication via $\mathcal{N}^{AE\rightarrow B}$ assisted by
  arbitrary pre-shared entanglement is known, then the entanglement-
  and helper-assisted quantum capacity is given by the largest
  $R$ such that the pair $(R,0)$ is achievable.
\end{remark}


\section{Conclusion}
\label{sec:conclusion}
\IEEEPARstart{W}{e} have laid the foundations of a theory of
quantum communication with passive environment-assistance, where
a helper is able to select the initial environment state of the 
channel, modelled as a unitary interaction. The general, multi-letter,
capacity formulas we gave for the quantum capacity assisted by an
unrestricted, and by a separable helper resemble the analogous
formula for the unassisted capacity. Like the latter, which is
contained as a special case, the environment-assisted
capacities are continuous in the channel, but in general seem to
be hard to characterize in simple ways.

In our development we have then focused on two-qubit unitaries, 
giving rise to very simple-looking qubit channels for which the
environment-assisted quantum capacity with separable helper can be
evaluated. Interestingly, there are unitaries giving rise to 
anti-degradable channels for every input state, hence the separable
helper capacity vanishes; yet, some of these ``universally anti-degradable''
unitaries could be super-activated by unitaries from the same class,
in some cases by themselves. 
In fact, there is a single unitary $\sqrt{\SWAP}$  
that activates all universally anti-degradable unitaries $U\in\mathfrak{A}$ 
(except itself, according to numerics).
In particular, the quantum capacity $Q_H$ with unrestricted helper can be 
strictly larger than the one with separable helper, $Q_{H\ox}$, and the 
computation of the former remains a major open problem. 

Some other interesting open questions include the following:
\begin{itemize}
  \item How to characterize the set of unitaries $U$ such that $Q_H(U) = 0$?
    Note that in the two-qubit case we only the example $U=\SWAP$, but 
    it seems that $\sqrt{\SWAP}$ is another one, but we lack a proof.
  \item Can $Q_H$ be super-activated, i.e.~are there $U$, $V$ with $Q_H(U)=Q_H(V)=0$
    but $Q_H(U\ox V)>0$? From the above analysis, $U=\SWAP$ and
    $V=\sqrt{\SWAP}$ seem good candidates
\end{itemize}

Finally, we only just started the issue of entangled environment-assistance,
motivated by the distinguished role of the SWAP gate in many of
our examples. But for the moment we do not even have an understanding
of additivity or super-activation of the entangled-helper assisted
capacities $Q_{EH}$ and $Q_{EH\ox}$.

\medskip
Looking further afield, our model and approach can evidently be
adapted to other communication capacities, say for instance
the private capacity $P$ and classical capacity $C$ of a channel.
Regarding the former, our examples of super-activation and 
self-super-activation apply directly because private and quantum
capacity coincide for degradable and anti-degradable channels.
On the classical capacity we have preliminary results which will be
reported on in forthcoming work~\cite{KM14}.


\section*{Acknowledgements}
SK thanks the 
Universitat Aut\`{o}noma de Barcelona for kind hospitality.
AW's work was supported
by the European Commission (STREPs ``QCS'' and ``RAQUEL''),
the European Research Council (Advanced Grant ``IRQUAT'') 
and the Spanish MINECO (grant FIS2008-01236) with FEDER funds.
DY's work was supported by the ERC (Advanced Grant ``IRQUAT'') and the 
NSFC (Grant No.~11375165).
Part of the work was done during the programme ``Mathematical Challenges in Quantum Information'' 
(MQI) at the Issac Newton Institute in Cambridge whose hospitality was gratefully acknowledged, 
where DY was supported by the Microsoft Visiting Fellowship. 
The authors thank Stefan B\"auml, Jan Bouda, Marcus Huber and Claude Kl\"ockl
for discussions on super-activation.



\appendices

\section{Communication in the presence of a jammer (QAVC)}
\label{app:jammer}
The purpose of this appendix is to prove the adversarial channel 
capacity theorem, which we restate here:

\textit{Theorem~\ref{thm:QJ}:}
  For any jammer channel $\mathcal{N}:{AE \rightarrow B}$, 
  \[
    Q_{J,r}(\mathcal{N}) 
      = \sup_n \max_{\rho^{(n)}} \min_{\eta} \frac1n I_c\bigl(\rho^{(n)};(\mathcal{N}_\eta)^{\ox n}\bigr),
  \]
  where the maximization is over states $\rho^{(n)}$ on $A^n$, and the minimization
  is over arbitrary states $\eta$ on $E$.

\medskip
\begin{proof}
The converse part, i.e.~the ``$\leq$'' inequality, follows from~\cite[Thm.~27]{ABBN},
because in the proof it is enough to consider tensor product strategies
$\eta^{(n)} = \eta_1\ox\cdots\ox\eta_n$ of the jammer, hence
$\mathcal{N}_{\eta^{(n)}} = \mathcal{N}_{\eta_1}\ox\cdots\ox\mathcal{N}_{\eta_n}$
is a tensor product map as in the AVQC model. Thus the proof of~\cite{ABBN}
applies unchanged.

For the direct part (``$\geq$''), consider input states $\rho^{(n)}$ on $A^n$
and a rate
\[
  R \leq \min_\eta I_c\bigl( \rho^{(n)};(\mathcal{N}_\eta)^{\ox n} \bigr) - \delta,
\]
for $\delta > 0$ and all integers $n$.
We invoke a result of Bjelakovi\'{c} \emph{et al.}~\cite{BBN09} on the
so-called \emph{compound channel} 
$\bigl((\mathcal{N}_\eta)^{\ox n}\bigr)_{\eta\in\mathcal{S}(E)}$, to the
effect that there exist codes $(\mathcal{D}_n,\mathcal{E}_n)$ for all
block lengths $n$ and with rate $R$ that perform universally well for all
the i.i.d.~channels $(\mathcal{N}_\eta)^{\ox n}$:
\[
  F_n := \min_\eta 
            F\bigl( \Phi^{RB_0},
                    (\mathcal{D}_n\circ\mathcal{N}_\eta^{\ox n}\circ\mathcal{E}_n)\Phi^{RA_0} \bigr)
      \geq 1-c^n,
\]
with some $c<1$. For later use, let us rephrase this condition as a
property of $\eta^{(n)} = \eta^{\ox n}$:
\begin{equation}\begin{split}
  c^n &\geq 1-F \\
      &= \tr (\1-\Phi)\bigl( \mathcal{D}_n\circ\mathcal{N}^{\ox n}(\mathcal{E}_n(\Phi)\ox\eta^{(n)})\bigr) \\
      &= \tr \bigl(({\mathcal{N}^\dagger}^{\ox n}\circ\mathcal{D}_n^\dagger)(\1-\Phi)\bigr)
             \bigl(\mathcal{E}_n(\Phi)\ox\eta^{(n)}\bigr)                                                  \\
      &= \tr X_n \eta^{(n)},
  \label{eq:F-as-expectation}
\end{split}\end{equation}
where $0\leq X_n\leq \1$ is a constant operator depending only on the code.

We claim that, using a shared uniformly random permutation
$\pi \in S_n$ to permute the $n$ input/output systems, the same code
is good against the jammer. Concretely, let $\mathcal{U}^\pi$ be the
conjugation by the permutation unitary on an $n$-party system, and
define, for a given $n$,
\begin{align*}
  \mathcal{E}_\pi &:= \mathcal{U}^\pi \circ \mathcal{E}_n, \\
  \mathcal{D}_\pi &:= \mathcal{D}_n \circ \mathcal{U}^{\pi^{-1}}.
\end{align*}
Then, for any jammer strategy $\eta^{(n)} \in \mathcal{S}(E^n)$,
\begin{align}
  1-&\overline{F}\left(\eta^{(n)}\right) \nonumber\\
    &= \frac{1}{n!} \sum_{\pi\in S_n} 
         1- F\bigl( \Phi^{RB_0},
                    \bigl(\mathcal{D}_\pi\circ(\mathcal{N}^{\ox n})_{\eta^{(n)}}\circ\mathcal{E}_\pi\bigr)
                                                                                        \Phi^{RA_0} \bigr) \nonumber\\
    &= \tr \left( X_n \frac{1}{n!} \sum_{\pi\in S_n} \mathcal{U}^{\pi}(\eta^{(n)}) \right) \nonumber\\
    &= \tr X_n \overline{\eta}^{(n)},
  \label{eq:Sn-symm-F}
\end{align}
using Eq.~(\ref{eq:F-as-expectation}), and where 
$\overline{\eta}^{(n)} = \frac{1}{n!} \sum_{\pi\in S_n} \mathcal{U}^{\pi}(\eta^{(n)})$
is permutation symmetric.

At this point, we can apply the postselection technique of~\cite{CKR}, which
relies on the matrix inequality
\[
  \overline{\eta}^{(n)} \leq (n+1)^{|E|^2} \int_\sigma {\rm d}\sigma \, \sigma^{\ox n},
\]
with a certain universal probability measure ${\rm d}\sigma$ over states on $E$.
Thus, according to the assumption and the above Eq.~(\ref{eq:Sn-symm-F}),
we find that for the permutation-symmetrized compound channel code,
\[
  1-\overline{F}(\eta^{(n)}) \leq (n+1)^{|E|^2} c^n
\]
for all jammer strategies $\eta^{(n)}$, and the right hand side of course 
still goes to zero exponentially fast, concluding the proof.
\end{proof}

\section{Parametrization of two-qubit unitaries \protect\\ and degradability regions}
\label{app:two-qubit}
For the further analysis we require another analytical criterion
for anti-degradability:

\begin{lemma}[Myhr/L\"utkenhaus~\cite{ML09}]
  A qubit channel with qubit environment is anti-degradable if and
  only if $\lambda_{\max}(\rho_{RB}) \leq \lambda_{\max}(\rho_{B})$, where
  $\lambda_{\max}(X)$ is the maximum eigenvalue of a Hermitian matrix $X$. 
  Here $\rho_{RB}$ is the Choi matrix of the given qubit channel and 
  $\rho_{B}$ is the reduced state after tracing out the reference system $R$.
  \qed
  \label{Lemma:antdeg}
\end{lemma}

\medskip
Following the analysis in Section~\ref{sec:two-qubit}, we restrict our attention
to the parameter space $\mathfrak{T}$
of $(\alpha_{x},\alpha_{y},\alpha_{z})$ satisfying 
$\frac{\pi}{2} \geq \alpha_x \geq \alpha_y \geq \alpha_z \geq 0$, which forms 
a tetrahedron with vertices
$(0,0,0)$, $(\frac{\pi}{2},0,0)$, $(\frac{\pi}{2},\frac{\pi}{2},0)$
and $(\frac{\pi}{2},\frac{\pi}{2},\frac{\pi}{2})$.

Given a unitary $U(\alpha_{x},\alpha_{y},\alpha_{z})$ and an initial state
of the environment,
$\vert \xi \rangle = \cos(\frac{\theta}{2})\ket{0} + e^{i\varphi}\sin(\frac{\theta}{2})\ket{1}$,
where $\theta \in [0,\pi]$, $\varphi \in [0,2\pi)$, we evaluate the Choi matrix by 
inputting a maximally entangled state 
$\ket{\Phi} = \frac{1}{\sqrt{2}}\bigl(\vert 00 \rangle + \vert 11 \rangle\bigr)$.
Thus the output state is
$\ket{\Psi}^{RBF} = (\1^{R} \otimes U^{AE})\bigl( \ket{\Phi}^{RA} \otimes \ket{\xi}^E \bigr)$. 
From the Schmidt decomposition, the maximum eigenvalue of $\rho_{RB}$ is equal to 
the maximum eigenvalue of
$\rho^{F} =\Tr_{RB} \proj{\Psi}$, which can be written in matrix form as
\ba
\frac12
\left[\begin{array}{cc}
  {1 + a_{F}}      & {b_{F} - ic_{F}} \\
  {b_{F} + ic_{F}} & {1 - a_{F}}
\end{array}\right],
\ea
with the Bloch vector components given by
\begin{align*}
  a_{F} &= \cos(\theta)\cos(\alpha_{x})\cos(\alpha_{y}),              \\
  b_{F} &= \sin(\theta)\cos(\varphi)\cos(\alpha_{z})\cos(\alpha_{y}), \\
  c_{F} &= \sin(\theta)\sin(\varphi)\cos(\alpha_{z})\cos(\alpha_{x}).
\end{align*}
Similarly, $\rho^{B} =\Tr_{RF} \proj{\Psi}$ has Bloch vector components given by
\begin{align*}
  a_{B} &= \cos(\theta)\sin(\alpha_{x})\sin(\alpha_{y}),              \\
  b_{B} &= \sin(\theta)\cos(\varphi)\sin(\alpha_{z})\sin(\alpha_{y}), \\
  c_{B} &= \sin(\theta)\sin(\varphi)\sin(\alpha_{z})\sin(\alpha_{x}).
\end{align*}

The largest eigenvalue of a qubit density matrix $\rho$ with Bloch vector components
$a,b,c$ is $\frac{1+\sqrt{a^2 + b^2 + c^2}}{2}$. When we impose the condition for 
anti-degradability from Lemma~\ref{Lemma:antdeg} we get the following inequality:
\[\begin{split}
  0 &\geq \cos^2(\theta) \cos(\alpha_{x} + \alpha_{y}) \cos(\alpha_{x} - \alpha_{y}) \\
    &\phantom{=}
     + \sin^2(\theta) \cos^2(\varphi) \cos(\alpha_{z} + \alpha_{y}) \cos(\alpha_{z} - \alpha_{y}) \\
    &\phantom{=}
     + \sin^2(\theta) \sin^2(\varphi) \cos(\alpha_{z} + \alpha_{x}) \cos(\alpha_{z} - \alpha_{x}).
\end{split}\]

This must be true for all input states of environment, hence for all $\theta \in [0,\pi]$,
$\varphi \in [0,2\pi)$. Thus we arrive at
\begin{equation}
  \label{UP}
  \alpha_{x} + \alpha_{y},\ 
  \alpha_{y} + \alpha_{z},\ 
  \alpha_{z} + \alpha_{x} \geq \frac{\pi}{2},
\end{equation}
for the universally anti-degradable region. This forms another tetrahedron with
vertices $(\frac{\pi}{4},\frac{\pi}{4},\frac{\pi}{4})$, 
$(\frac{\pi}{2},\frac{\pi}{4},\frac{\pi}{4})$,
$(\frac{\pi}{2},\frac{\pi}{2},0)$ and $(\frac{\pi}{2},\frac{\pi}{2},\frac{\pi}{2})$, 
which is depicted in Fig.~\ref{fig:uad}.

By swapping the outputs of unitary $U \in \mathfrak{A}$ we get another 
unitary $V = \SWAP \cdot U \in \mathfrak{D}$. By applying this transformation to 
the vertices of the parameter region of $\mathfrak{A}$,
we get the vertices of the parameter region $\mathfrak{D}$ given by 
$(\frac{\pi}{4},\frac{\pi}{4},\frac{\pi}{4})$,
$(\frac{\pi}{4},\frac{\pi}{4},0)$, $(\frac{\pi}{2},0,0)$ and $(0,0,0)$. 
The unitary $\sqrt{\SWAP}$, with the parameters $(\frac{\pi}{4},\frac{\pi}{4},\frac{\pi}{4})$,
is the unique unitary which lies in the intersection of $\mathfrak{A}$ and $\mathfrak{D}$. 
This gives rise to symmetric qubit channels for every initial state of the environment.



\begin{thebibliography}{99}
\bibitem{ABBN} 
R. Ahlswede, I. Bjelakovi\'{c}, H. Boche, and J. N\"otzel,
``Quantum capacity under adversarial noise: arbitrarily varying quantum channels'', 
\textit{Communications in Mathematical Physics}, vol. 317(1), pp. 103-156, 2013.

\bibitem{BNS98}
H. N. Barnum, M. A. Nielsen, and B. Schumacher,
``Information transmission through a noisy quantum channel'',
\textit{Physical Review A}, vol. 57(6), pp. 4153-4175, 1998.

\bibitem{BDS97}
C. H. Bennett, D. P. DiVincenzo, and J. A. Smolin, ``Capacities of quantum erasure channels'', 
\textit{Physical Review Letters}, vol. 78, pp. 3217-3220, 1997.

\bibitem{BBN09} 
I. Bjelakovi\'{c}, H. Boche, and J. N\"otzel, 
``Entanglement transmission and generation under channel uncertainty: 
Universal quantum channel coding'', 
\textit{Communications in Mathematical Physics}, vol. 292, pp. 55-97, 2009.

\bibitem{BN13} 
H. Boche and J. N\"otzel,
``Arbitrarily Small Amounts of Correlation for Arbitrarily Varying Quantum Channels'',
{arXiv[quant-ph]:1301.6063}, 2013.

\bibitem{BN14}
H. Boche and J. N\"otzel, ``Positivity, Discontinuity, Finite Resources and Nonzero Error for
Arbitrarily Varying Quantum Channels'',
{arXiv[quant-ph]:1401.5360}, 2014.

\bibitem{BEHY11} 
F. G. S. L. Brand\~{a}o, J. Eisert, M. Horodecki, and D. Yang,
``Entangled inputs cannot make imperfect quantum channels perfect'',
\textit{Physical Review Letters}, vol. 106, pp. 230502, 2011.

\bibitem{BCD05}
F. Buscemi, G. Chiribella, and G. M. D'Ariano, 
``Inverting Quantum Decoherence by Classical Feedback from the Environment'',
\textit{Physical Review Letters}, vol. 95, pp. 090501, 2005.

\bibitem{CKR}
M. Christandl, R. K\"onig, and R. Renner,
``Postselection Technique for Quantum Channels with Applications to Quantum Cryptography'',
\textit{Physical Review Letters}, vol. 102, pp. 020504, 2009.

\bibitem{DBF07}
A. D'Arrigo, G. Benenti, and G. Falci,
``Quantum Capacity of a dephasing channel with memory'', 
\textit{New Journal of Physics}, vol.~9(9), pp.~310, 2007.

\bibitem{Devetak03}
I. Devetak, ``The Private Classical Capacity and Quantum Capacity of a Quantum Channel'', 
\textit{IEEE Transactions on Information Theory}, vol. 51(1), pp. 44-55, 2005.

\bibitem{DS05}
I. Devetak and P. W. Shor,
``The Capacity of a Quantum Channel for Simultaneous 
Transmission of Classical and Quantum Information'',
\textit{Communications in Mathematical Physics}, vol. 256(2), pp. 287-303, 2005.

\bibitem{DSS98}
D. P. DiVincenzo, P. W. Shor, and J. A. Smolin,
``Quantum-channel capacity of very noisy channels'', 
\textit{Physical Review A}, vol. 57(2), pp. 830-839, 1998.

\bibitem{Fannes} M. Fannes,
``A Continuity Property of the Entropy Density for Spin Lattice Systems'',
\textit{Communications in Mathematical Physics}, vol. 31, pp. 291-294, 1973.

\bibitem{Fuchs-vandeGraaf}
C. A. Fuchs and J. van de Graaf,
``Cryptographic Distinguishability Measures for Quantum-Mechanical States'',
\textit{IEEE Transactions on Information Theory}, vol. 45(4), pp. 1216-1227, 1999.

\bibitem{GW03}
M. Gregoratti and R. F. Werner, ``Quantum lost and found'', 
\textit{Journal of Modern Optics}, vol. 50(6-7), pp. 915-933, 2003.

\bibitem{GW04}
M. Gregoratti and R. F. Werner, 
``On quantum error-correction by classical feedback in discrete time'', 
\textit{Journal of Mathematical Physics}, vol. 45(7), pp. 2600-2612, 2004.


\bibitem{HVC02}
K. Hammerer, G. Vidal, and J. I. Cirac,
``Characterization of nonlocal gates",
\textit{Physical Review A}, vol. 66(6), pp. 062321, 2002.

\bibitem{HK05}
P. Hayden and C. King, 
``Correcting quantum channels by measuring the environment'', 
\textit{Quantum Information and Computation}, vol. 5(2), pp. 156-160, 2005.

\bibitem{HW97}
S. Hill and W. K. Wootters, 
``Entanglement of a Pair of Quantum Bits",
\textit{Physical Review Letters}, vol. 78(26), pp. 5022-5025, 1997.

\bibitem{HDW08}
M.-H. Hsieh, I. Devetak, and A. Winter, 
``Entanglement-Assisted Capacity of Quantum Multiple-Access Channels'',
\textit{IEEE Transactions on Information Theory}, vol. 54(7), pp. 3078-3090, 2008.

\bibitem{KM14} 
S. Karumanchi, S. Mancini, and A. Winter, in preparation, 2014.

\bibitem{KR01}
C. King and M. B. Ruskai, 
``Minimal Entropy of States Emerging from Noisy Quantum Channels'', 
\textit{IEEE Transactions on Information Theory}, vol. 47, pp. 192-209, 2001.

\bibitem{KC01}
B. Kraus and J. I. Cirac, 
``Optimal creation of entanglement using a two-qubit gate'', 
\textit{Physical Review A}, vol. 63, pp. 062309, 2001.

\bibitem{LGZ13} 
Y. Liu, Y. Guo, and D. L. Zhou,
``Optimal transfer of an unknown state via a bipartite quantum operation'',
\emph{Europhysics Letters}, vol. 102(5), pp. 50003, 2013.

\bibitem{LPM09}
C. Lupo, O. Pilyavets, and S. Mancini,
`` On the capacities of a lossy bosonic channel with correlated noise'',
\textit{New Journal of Physics}, vol.~11, pp.~063023, 2009

\bibitem{LS09}
D. Leung and G. Smith, ``Continuity of Quantum Channel Capacities'', 
\textit{Communications in Mathematical Physics}, vol. 292, pp. 201-215, 2009.

\vfill\pagebreak

\bibitem{Lloyd96}
S. Lloyd, ``Capacity of the noisy quantum channel'', 
\textit{Physical Review A}, vol. 55(3), pp. 1613-1622, 1996.

\bibitem{MCM11}
L. Memarzadeh, C. Cafaro, and S. Mancini,
``Quantum information reclaiming after amplitude damping'',
\textit{Journal of Physics A: Mathematical and Theoretical},
vol. 44, pp. 045304, 2011.

\bibitem{MMM11}
L. Memarzadeh, C. Macchiavello, and S. Mancini,
``Recovering quantum information through partial access to the environment'',
\textit{New Journal of Physics}, vol. 13, pp. 103031, 2011. 

\bibitem{ML09}
G. O. Myhr and N. L\"utkenhaus,
``Spectrum conditions for symmetric extendible states'', 
\textit{Physical Review A}, vol. 79, pp. 062307, 2009. 

\bibitem{RSW02}
M. B. Ruskai, S. Szarek, and E. Werner, 
``An Analysis of Completely-Positive Trace-Preserving Maps on 2x2 Matrices'', 
\textit{Linear Algebra and Its Applications}, vol. 347, pp. 159-187, 2002.

\bibitem{Schu96} B. Schumacher,
``Sending entanglement through noisy quantum channels'',
\textit{Physical Review A}, vol. 54(4), pp. 2614-2628, 1996.

\bibitem{NS96}
B. Schumacher and M. A. Nielsen,
``Quantum data processing and error correction'',
\textit{Physical Review A}, vol. 54(4), pp. 2629-2635, 1996. 

\bibitem{SS96}
P. W. Shor and J. A. Smolin, 
``Quantum Error-Correcting Codes Need Not Completely Reveal the Error Syndrome'', 
arXiv:quant-ph/9604006, 1996.

\bibitem{Shor00}
P. W. Shor, ``The quantum channel capacity and coherent information'', 
MSRI seminar, November 2002.

\bibitem{SY08}
G. Smith and J. Yard,
``Quantum Communication with Zero-Capacity Channels'',
\textit{Science}, vol. 321(5897), pp. 1812-1815, 2008.

\bibitem{SVW05}
J. A. Smolin, F. Verstraete, and A. Winter, 
``Entanglement of assistance and multipartite state distillation'', 
\textit{Physical Review A}, vol. 72, pp. 052317, 2005. 

\bibitem{Wilde11}
M. M. Wilde, \emph{Quantum Information Theory}, Cambridge University Press, Cambridge 2013;
``From Classical to Quantum Shannon Theory'', arXiv[quant-ph]:1106.1445.

\bibitem{W01}
A. Winter, ``The capacity of the quantum multiple access channel'',
\textit{IEEE Transactions on Information Theory}, vol. 47(7), pp. 3059-3065, 2001.

\bibitem{Winter07}
A. Winter, ``On Environment-Assisted Capacities of Quantum Channels'', 
\textit{Markov Processes and Related Fields}, vol. 13(1-2), pp. 297-314, 2007. 

\bibitem{WPG07}
M. M. Wolf and D. Perez-Garc\'ia, 
``Quantum capacities of channels with small environment'', 
\textit{Physical Review A}, vol. 75, pp. 012303, 2007. 

\bibitem{YHD08}
J. Yard, P. Hayden, and I. Devetak,
``Capacity Theorems for Quantum Multiple Access Channels: 
Classical-Quantum and Quantum-Quantum Capacity Regions'',
\textit{IEEE Transactions on Information Theory}, vol. 54(7), pp. 3091-3113, 2008.





\end{thebibliography}
\end{document}